\documentclass{amsart}

\usepackage{enumerate}
\usepackage{amssymb}
\usepackage{amsfonts}
\usepackage{latexsym}
\usepackage{amsmath}
\usepackage{euscript}
\usepackage{graphics}
\usepackage{graphicx}
\usepackage{tikz}
\usetikzlibrary{spy}
\graphicspath{ {./images/} }

\newtheorem{theorem}{Theorem}[section]
\newtheorem{proposition}[theorem]{Proposition}
\newtheorem{corollary}[theorem]{Corollary}
\newtheorem{lemma}[theorem]{Lemma}
\newtheorem{definition}[theorem]{Definition}
\theoremstyle{remark}
\newtheorem{example}[theorem]{Example}
\newtheorem{remark}[theorem]{Remark}

\numberwithin{equation}{section}

\def\Im{\operatorname{Im}}

\usepackage[colorinlistoftodos]{todonotes}

\newcommand{\tc}{\textcolor{black}}

\newcommand{\tcnew}{\textcolor{black}}

\begin{document}
\title[DEK-Type polynomials and the Christoffel formula ]{\tcnew{DEK-Type orthogonal} polynomials and a modification of the Christoffel formula }

\author[R.~Bailey]{Rachel~Bailey}
\address{
RB,
Department of Mathematics\\
University of Connecticut\\
341 Mansfield Road, U-1009\\
Storrs, CT 06269-1009, USA}
\email{rachel.bailey@uconn.edu}

\author[M.~Derevyagin]{Maxim~Derevyagin}
\address{
MD,
Department of Mathematics\\
University of Connecticut\\
341 Mansfield Road, U-1009\\
Storrs, CT 06269-1009, USA}
\email{maksym.derevyagin@uconn.edu}


\begin{abstract}
\tc{In this note we revisit one of the first known examples of exceptional orthogonal polynomials that was
 introduced by Dubov, Eleonskii, and Kulagin in relation to nonharmonic oscillators with equidistant spectra. We dissect the DEK polynomials using the discrete Darboux transformations and unravel a characterization bypassing the differential equation that defines the DEK polynomials. This characterization also leads to a family of \tcnew{general orthogonal polynomials with missing degrees} and this approach manifests its relation to biorthogonal polynomials introduced by Iserles and N\o rsett, which are applicable to a whole range of problems in computational and applied analysis. We also obtain a modification of the Christoffel formula for this family since
its classical form cannot be applied in this case.}
\end{abstract}

\maketitle

\section{Introduction}

The classical orthogonal polynomials have shown themselves to be very useful in a wide range of various branches of mathematics. One of the reasons is that they satisfy both differential and difference equations. This naturally led to a separate industry that was concerned with the question on how to obtain more families of polynomials or functions that have those bispectral properties. For instance, Reach \cite{R88} showed that Darboux transformations applied to a differential operator, whose eigenfunctions satisfy a recurrence relation, leads to a new family that satisfy both differential and difference equations. Oftentimes, one encounters the problem from a different perspective: a certain perturbation of a problem to which one applies classical orthogonal polynomials would lead to a new family of polynomials which would satisfy a differential equation or a difference equation, or both.

To demonstrate how one can encounter new families, let us recall that an example of a potential of an anharmonic oscillator such that  the Hamiltonian operator has a strictly equidistant part of the spectrum that corresponds to all the excited states was given in \cite{D92}. This potential gives rise to the monic polynomials $F_n(x)$, which we refer to as the DEK polynomials, defined by the differential equation
\begin{equation}\label{DEforF}
 (1+x^2)\left(\frac{d^2F_n(x)}{dx^2}-x\frac{dF_n(x)}{dx}+(n+2)F_n(x)\right)=4x\frac{dF_n(x)}{dx},   
\end{equation}
where $n=1,2,3,\dots$ and hence $\deg F_n=n+2$. Equation \eqref{DEforF} also has a constant solution when $n=-2$ and for consistency, we set $F_0(x)=1$. Notably, $F_0(x)$ corresponds to the ground state of the system but the gap separates this state from the first excited state that corresponds to $F_1(x)=x^3+3x$. This construction was further investigated and generalized in \cite{A}, \cite{BS95}, \cite{DEK94}, \cite{SO95A}, and \cite{SO95B}. In particular, a connection to  Darboux transformations of the differential equation that defines Hermite polynomials was explicitly given in \cite{A} and its relation to commutation methods was pointed out in \cite{BS95}, \cite{SO95B}. The polynomials $F_n(x)$ can also be expressed via Rodrigues' formula \cite{D92}
\begin{equation}
    F_n(x)=(-1)^n(1+x^2)^2e^{\frac{x^2}{2}}\frac{d}{dx}\left((1+x^2)^{-1}\frac{d^{n-1}}{dx^{n-1}}e^{-\frac{x^2}{2}}\right), \quad n=1,2,3,\dots,
\end{equation}
which, in turn, can be used to prove the orthogonality relation \cite{D92}
\begin{equation}\label{DEKorthogonality}
 \int_{-\infty}^\infty F_n(x)F_m(x) \frac{e^{-\frac{x^2}{2}}}{(1+x^2)^2}dx =(n-1)(n-1)!(2\pi)^\frac{1}{2}\delta_{n,m}   
\end{equation}
for any nonnegative integers $n$ and $m$, where $\delta_{n,m}$ is the Kronecker delta. At the same time, the polynomials $F_n(x)$ are closely connected to the monic Hermite polynomials
\[
He_n(x)=(-1)^ne^{\frac{x^2}{2}}\frac{d^n}{dx^n}e^{-\frac{x^2}{2}},
\quad n=0,1,2,3,\dots,
\]
where the latter is known to satisfy the three-term recurrence relation
\begin{equation}\label{HermiteRec}
 xHe_n(x)=He_{n+1}(x)+nHe_{n-1}(x), \quad n=1,2,3,\dots.   
\end{equation}
More precisely, the relation between the sequences of Hermite and DEK polynomials is given by the formula (for instance, see \cite{A})
\begin{equation}\label{FviaWronskii}
\begin{split}
F_{n}(x)=&\frac{1}{n(n+1)}\begin{vmatrix}
He_1(x)&He_2(x)&He_{n+2}(x)\\
He_1'(x)&He_2'(x)&He_{n+2}'(x)\\
He_1''(x)&He_2''(x)&He_{n+2}''(x)
\end{vmatrix}\\
=&\frac{1}{n(n+1)}\begin{vmatrix}
x&x^2-1&He_{n+2}(x)\\
1&2x&He_{n+2}'(x)\\
0&2&He_{n+2}''(x)
\end{vmatrix}, \quad n=1,2,3,\dots, 
\end{split}
\end{equation}
which means that the polynomials $F_n(x)$ correspond to (continuous) Darboux transformation (for more information about continuous Darboux transformation, see \cite{A}). As a consequence, the results of \cite{R88} can be applied in this case and thus we know that $F_n$'s satisfy a recurrence relation \tcnew{(see \cite{GKKM} where it is shown that exceptional Hermite polynomials, which include $F_n$'s as a particular case, satisfy multiple recurrence relations)}. Although \eqref{FviaWronskii} was not given in \cite{D92} explicitly, the relation
\begin{equation}\label{DEKequationForF}
 F_n(x)=(x^3+3x)He_{n-1}(x)-(n-1)(1+x^2)He_{n-2}(x), \quad n=1,2,\dots
 \end{equation}
 was proved and it is equivalent to \eqref{FviaWronskii} through \eqref{HermiteRec}. Note that applying \eqref{HermiteRec} to \eqref{DEKequationForF} one can get \cite{D92}
 \begin{equation}\label{GeronimusTr}
 F_n(x)=He_{n+2}(x)+2(n+2)He_{n}(x)+(n+2)(n-1)He_{n-2}(x),   
\end{equation}
which shows that $F_n(x)$ is a linear combination of 3 Hermite polynomials. Actually, formula \eqref{GeronimusTr} suggests that the family of the polynomials $F_n(x)$ might be a discrete Darboux transformation of Hermite polynomials $He_n(x)$ (for the definition and basic properties of discrete Darboux transformations see \cite{BM04}, \cite{GH97}, \cite{MatS79}), which will be confirmed and used in this paper. It is also worth noting here that the above-described construction was recently given a new flavor and farther generalized to new algebraic and spectral theory levels (see the recent papers \cite{D14}, \tcnew{\cite{D21}}, \cite{GGM19}, \tcnew{\cite{GGM}}, \tcnew{\cite{GM20}}, \cite{KM20}, \cite{LLK15}, \cite{STZ10} and references therein). In particular, in \cite{D14} the construction of Meixner orthogonal polynomials was presented, which already put exceptional orthogonal polynomials outside of the context of differential equations.

On the other hand, the theory of orthogonal polynomials is not restricted to just classical orthogonal polynomials and these days general orthogonal polynomials are even more important than the classical ones due to the development of computational mathematics and spectral theory to name a few. For example, general orthogonal polynomials appear as denominators of rational approximants that are called Pad\'e approximants \cite{BG}. In some cases for some degrees Pad\'e approximants might not exist, which leads to some gaps in the corresponding sequence of orthogonal polynomials (see \cite{BG}). Although it may seem unrelated to the polynomials $F_n$ with missing degrees 1 and 2, having seen these two occurrences side by side {\it it is natural to ask if gaps in exceptional orthogonal polynomials and gaps in Pad\'e approximants have the same nature. Findings of this paper demonstrate that the answer to this question is affirmative and at the same time the approach puts exceptional orthogonal polynomials in the framework of discrete Darboux transformations as was done for other nonstandard orthogonalities}. Although from the modern point of view, the DEK polynomials $F_n(x)$ are a particular case of exceptional Hermite polynomials, the transparent construction of the family provides a certain insight into the theory of exceptional orthogonal polynomials that we exploit in this paper. In most of the cases, our approach is not restricted to the DEK type polynomials and one can adapt our findings to the general case of exceptional Hermite polynomials using the already developed machinery. 

Now we are in the position to briefly describe the structure of the paper. To this end observe that one can deduce from \eqref{DEforF} that 
\begin{equation}\label{FzeoroCond}
    F'_n(i)=F'_n(-i)=0,
\end{equation}
which is an interpolation condition but it will be shown in Section 3 that it can be thought of as a part of biorthogonality, the concept that was introduced in \cite{IN88} and generalized in \cite{B89}, \cite{B92}, and that appears when solving various problems \cite{LS22}, \cite{IN87}, \cite{IS87}. Before that, in Section 2, we will demonstrate that \eqref{GeronimusTr} and \eqref{FzeoroCond} are characteristic for a class of orthogonal polynomials with missing degrees 1 and 2 that includes the DEK polynomials, which is why we will call them \tcnew{DEK-type orthogonal} polynomials. Then, in Section 5, we recast this class of polynomials as a specific multiple discrete Darboux transformation (cf. \cite{DD10}, \cite{DGM14},\cite{GH97}, \cite{GHH99}; note that in \cite{DD10} it is shown that single step discrete Darboux transformations lead to families of orthogonal polynomials with gaps). To do that, in Section 4, we will introduce a modification of the Christoffel formula that works for the exceptional orthogonal polynomials in question since the classical form of the Christoffel formula cannot be applied. At the end, we will show that this construction is applicable to the Chebyshev polynomials and as a result we will present a new family of \tcnew{DEK-type orthogonal} polynomials related to the Chebyshev polynomials.

\section{\tcnew{DEK-Type Orthogonal} Polynomials}
 
 In this section we present the general construction of \tcnew{DEK-type orthogonal} polynomials starting with a family of symmetric orthogonal polynomials. 
 
 Let \tcnew{$\{P_n(x)\}_{n\geq 0}$} be a  sequence of monic orthogonal polynomials with respect to a symmetric measure $\mu$ \tcnew{supported on a symmetric subset of the real line}. 
 We will consider a sequence of polynomials $R_n(x)$ defined as \begin{equation}\label{eq:Rn} R_n(x):=P_{n+2}(x)+A_nP_n(x)+B_nP_{n-2}(x), \quad R_0(x) := 1,\end{equation} 
for sequences $\{A_n\}$ and $\{B_n\}$ of real numbers such that

\begin{equation}\label{eq:Rn'(i)=0}
    R'_n(i)=0 \quad n=1,2,\dots
    \end{equation}
    \begin{equation}\label{eq:integral_R0}
     \int_{-\infty}^{\infty}  R_0(x)R_n(x)\frac{d\mu(x)}{(1+x^2)^2}=0 \text{ for all } n=1,2,\dots
     \end{equation} and
     \begin{equation}\label{eq:integral_R1}
     \int_{-\infty}^{\infty}  R_1(x)R_n(x)\frac{d\mu(x)}{(1+x^2)^2}=0 \quad n \geq 2.
\end{equation}

\begin{remark}\label{rmk:even_odd}Note that if $\mu$ is symmetric then $P_n(x)$ is an even (odd) function when $n$ is even (odd).
 \end{remark}
 \begin{proposition}\label{existence} Suppose $\{P_n(x)\}$ is a sequence of monic polynomials orthogonal with respect to a symmetric measure $\mu$. 
Let \[\mathcal{E}_n = \begin{pmatrix}
P'_n(i)&P'_{n-2}(i)\\
\int_{-\infty}^{\infty}   P_n(x) \frac{d\mu(x)}{(1+x^2)^2} & \int_{-\infty}^{\infty}  P_{n-2}(x) \frac{d\mu(x)}{(1+x^2)^2}
\end{pmatrix}\] and let \[\mathcal{O}_n= \begin{pmatrix}
P'_n(i)&P'_{n-2}(i)\\
\int_{-\infty}^{\infty}  R_1(x)P_n(x) \frac{d\mu(x)}{(1+x^2)^2} & \int_{-\infty}^{\infty} R_1(x)P_{n-2}(x) \frac{d\mu(x)}{(1+x^2)^2}
\end{pmatrix}.\] Then there exists a sequence of polynomials $\{R_n(x)\}$ defined as in equation \eqref{eq:Rn} which satisfy  \eqref{eq:Rn'(i)=0}, \eqref{eq:integral_R0} and \eqref{eq:integral_R1} above if and only if for all $k=1,2,\dots$, we have that $\det{\mathcal{E}_n}\neq 0$ for $n=2k$, and  $\det{\mathcal{O}_n}\neq 0$ for $n= 2k+1$.

 \end{proposition}\label{prop:existence} 
 \begin{proof}
First note that if $R_1(x)=P_3(x)+A_1P_1(x)$, then by remark \ref{rmk:even_odd},
$R_1(x)$ must be of the form $x^3+ax+A_1x$ for some real number $a$ where $P_3(x)=x^3+ax$. Then $R_1(x)$ satisfies condition \eqref{eq:integral_R0} since it is an odd function. In order for $R_1(x)$ to satisfy \eqref{eq:Rn'(i)=0}, we must have that $A_1+a=3$ and hence $R_1(x)=x^3+3x$ .

For $n \geq 2$, the proof follows simply by noting that conditions \eqref{eq:Rn'(i)=0} and \eqref{eq:integral_R0} are equivalent to the following system of equations:
\[\begin{cases}
   A_nP'_n(i)+B_nP'_{n-2}(i)=-P'_{n+2}(i) \\
   A_n\int_{-\infty}^{\infty}  P_n(x) \frac{d\mu(x)}{(1+x^2)^2} +B_n\int_{-\infty}^{\infty}  P_{n-2}(x) \frac{d\mu(x)}{(1+x^2)^2}= -\int_{-\infty}^{\infty}  P_{n+2}(x)\frac{d\mu(x)}{(1+x^2)^2}
   \end{cases}
   \]
 and conditions \eqref{eq:Rn'(i)=0} and \eqref{eq:integral_R1} are equivalent to:
 
\[
\begin{cases}
   A_nP'_n(i)+B_nP'_{n-2}(i)=-P'_{n+2}(i) \\ \scalebox{0.85}{
   $A_n\int_{-\infty}^{\infty}   R_1(x)P_n(x) \frac{d\mu(x)}{(1+x^2)^2} +B_n\int_{-\infty}^{\infty}  R_1(x)P_{n-2}(x) \frac{d\mu(x)}{(1+x^2)^2}= -\int_{-\infty}^{\infty}  R_1(x)P_{n+2}(x)\frac{d\mu(x)}{(1+x^2)^2}$.}
   \end{cases}
   \] 
 \end{proof}
 If such a family of polynomials $\tcnew{\{R_n(x)\}_{n\geq 0}}$ exists, then it must be the case that the family is orthogonal.

 \begin{theorem} \label{thrm:Rn}
The polynomials $R_n(x)$ are orthogonal with respect to $\frac{d\mu(x)}{(1+x^2)^2}$ i.e. 
\[\int_{-\infty}^{\infty}  R_n(x)R_m(x)\frac{d\mu(x)}{(1+x^2)^2} = 0
\] for $n \neq m$ and is nonzero for $m=n$.
\end{theorem}

\begin{proof}

First consider the case when $n$ and $m$ are both even and let $G_m(x):=R_m(x)-R_m(i)$. Then $G_m(i)=0$ and $G_m(-i)=R_m(-i)-R_m(i)=0$ since $R_m(x)$ is an even polynomial for $m$ even. Also, $G_m'(i)=G_m'(-i)=0$ by equation \eqref{eq:Rn'(i)=0}. Thus, for $m \geq 2$, we have that $G_m(x)=(1+x^2)^2g_m(x)$ where $g_m(x)$ is a polynomial of degree $m-2$.  Then for $m\geq 2$,
\begin{align*}
    \int_{-\infty}^{\infty}  \frac{G_m(x)R_n(x)}{(1+x^2)^2}d\mu(x) &= \int_{-\infty}^{\infty}  g_m(x)R_n(x)d\mu(x)\\
    &=\int_{-\infty}^{\infty}  g_m(x)P_{n+2}(x)d\mu(x) +A_n\int_{-\infty}^{\infty}  g_m(x)P_{n}(x)d\mu(x)\\
    &+B_n\int_{-\infty}^{\infty}  g_m(x)P_{n-2}(x)d\mu(x)\\
    &=0 \text{ if } m<n
\end{align*} where the last equality holds by the orthogonality of $\{P_n(x)\}$. But,
\begin{align*}
    \int_{-\infty}^{\infty}  \frac{G_m(x)R_n(x)}{(1+x^2)^2}d\mu(x)&= \int_{-\infty}^{\infty}  \frac{R_m(x)R_n(x)}{(1+x^2)^2}d\mu(x)- R_m(i)\int_{-\infty}^{\infty}  \frac{R_n(x)}{(1+x^2)^2}d\mu(x)\\
    &=\int_{-\infty}^{\infty}  \frac{R_m(x)R_n(x)}{(1+x^2)^2}d\mu(x)
\end{align*} by condition \eqref{eq:integral_R0} of the $\{R_n(x)\}$. Hence, we see that 
\[\int_{-\infty}^{\infty}  \frac{R_m(x)R_n(x)}{(1+x^2)^2}d\mu(x)=0\]
for all $m,n$ even such that $2\leq m <n$. If $m=0$ then it follows directly from \eqref{eq:integral_R0} that $R_n(x)$ is orthogonal to $R_0(x)$ for all $n \geq 2$.

Now, let $m,n$ be odd and consider $S_m(x)=x^4+\frac{R_m(i)}{2i}x^3+2x^2+\frac{3R_m(i)}{2i}x+1$. Then it is easy to check that for all $m \geq 1,$ 
\begin{enumerate}
    \item $S_m(i)=R_m(i)$
    \item $S_m(-i)=-R_m(i)=R_m(-i) \text{ (since } R_m(x) \text{ is odd)}$
    \item $S'_m(i)=S'_m(-i)=0$
\end{enumerate}
Let $H_m(x)=R_m(x)-S_m(x)$. Then $H_m(i)=H_m(-i)=H'_m(i)=H'_m(-i)=0$ for all $m \geq 1$, hence $H_m(x)=(1+x^2)^2h_m(x)$ where $h_m(x)$ is a polynomial of degree $m-2$. Then
\begin{equation}\label{eq:HmRn}
\begin{split}
    \int_{-\infty}^{\infty}  \frac{H_m(x)R_{n}(x)}{(1+x^2)^2}d\mu(x)&=\int_{-\infty}^{\infty}  \frac{h_m(x)P_{n+2}(x)}{(1+x^2)^2}d\mu(x)+A_n\int_{-\infty}^{\infty}  \frac{h_m(x)P_{n}(x)}{(1+x^2)^2}d\mu(x)\\
    &+B_n\int_{-\infty}^{\infty} \frac{h_m(x)P_{n-2}(x)}{(1+x^2)^2}d\mu(x)\\
    &=0 \text{ for } m<n.
    \end{split}
\end{equation} Also,
\begin{align*}
    \int_{-\infty}^{\infty}  \frac{H_m(x)R_{n}(x)}{(1+x^2)^2}d\mu(x)&=\int_{-\infty}^{\infty}  \frac{R_m(x)R_{n}(x)}{(1+x^2)^2}d\mu(x)-\int_{-\infty}^{\infty}  \frac{S_m(x)R_{n}(x)}{(1+x^2)^2}d\mu(x)
\end{align*} thus,
\begin{align*}
    \int_{-\infty}^{\infty}  \frac{R_m(x)R_{n}(x)}{(1+x^2)^2}d\mu(x)&=\int_{-\infty}^{\infty}  \frac{H_m(x)R_{n}(x)}{(1+x^2)^2}d\mu(x)+\int_{-\infty}^{\infty}  \frac{S_m(x)R_{n}(x)}{(1+x^2)^2}d\mu(x)\\
    &=\int_{-\infty}^{\infty}  \frac{S_m(x)R_{n}(x)}{(1+x^2)^2}d\mu(x)
\end{align*} for $3\leq m <n$ by \eqref{eq:HmRn}. 
Now, since $n$ is odd, 
\begin{equation}\label{eq:SmRn}
\begin{split}
    \int_{-\infty}^{\infty}  \frac{S_m(x)R_{n}(x)}{(1+x^2)^2}d\mu(x) &=\frac{R_m(i)}{2i}\int_{-\infty}^{\infty}  \frac{x^3R_{n}(x)}{(1+x^2)^2}d\mu(x)\\
    &+\frac{3R_m(i)}{2i}\int_{-\infty}^{\infty}  \frac{xR_{n}(x)}{(1+x^2)^2}d\mu(x)
    \end{split}
\end{equation} by definition of $S_m(x)$ and the fact that $\frac{x^kR_{n}(x)}{(1+x^2)^2}$ is an odd function for any positive, even integer $k$. But by condition \eqref{eq:integral_R1}, we know that 
\[\int_{-\infty}^{\infty}  \frac{R_1(x)R_n(x)}{(1+x^2)^2}d\mu(x)=\int_{-\infty}^{\infty}  \frac{(x^3+3x)R_n(x)}{(1+x^2)^2} d\mu(x)=0  \text{ for all }n \geq 2.
\] Hence, the right-hand side of \eqref{eq:SmRn} is zero for $n \geq 3$. 
Therefore, $\int_{-\infty}^{\infty}  \frac{R_m(x)R_{n}(x)}{(1+x^2)^2}d\mu(x)=0$ for all odd $m,n$ such that $3 \leq m <n$.

For the case where $m$ is even and $n$ is odd, or vice versa, $\int_{-\infty}^{\infty}  \frac{R_m(x)R_{n}(x)}{(1+x^2)^2}d\mu(x) =0$ simply by symmetry hence we have shown that $\int_{-\infty}^{\infty} \frac{R_m(x)R_{n}(x)}{(1+x^2)^2}d\mu(x)=0$ for all $0\leq m<n$.

Finally, notice that if $m=n$ then $\frac{R^2_n(x)}{(1+x^2)^2}$ is positive on the support of $\mu$ hence
\[\int_{-\infty}^{\infty}  \frac{R^2_n(x)}{(1+x^2)^2}d\mu(x) \neq 0
\]
thus $\tcnew{\{R_n(x)\}_{n\geq 0}}$ is a \tcnew{family of polynomials orthogonal} with respect to
\newline
$\frac{d\mu(x)}{(1+x^2)^2}$.
\end{proof}
 
If such a sequence $\{R_n(x)\}$ exists then $R_n(i)\neq 0$ for any $n \geq 0$. To show this, we will first prove the following lemma.

\begin{lemma}\label{lemma:zeros}
Given an arbitrary set $\{x_1,x_2,\dots,x_{n-1}\}$ of distinct real numbers, there is a monic polynomial $f(x)$ such that $\deg f=n+1$,
\[
f(x_j)=0,\quad j=1,2,\dots,n-1,
\]
where the $x_j$ are the only real zeros of $f$, and
\begin{equation}\label{iCond}
    f'(i)=f'(-i)=0.
\end{equation}
\end{lemma}
\begin{proof}
Evidently, if such a polynomial exists then $f(x)=P(x)Q(x)$, where 
$Q(x)=(x-x_1)(x-x_2)\dots(x-x_{n-1})$ is a real polynomial and $P(x)=x^2+bx+c$. In other words, the condition \eqref{iCond} is equivalent to
\begin{equation}\label{iSystem}
\begin{split}
    P'(i)Q(i)+P(i)Q'(i)=0\\
    P'(-i)Q(-i)+P(-i)Q'(-i)=0.
\end{split}
\end{equation}
Since $P'(i)=2i+b$, $P'(-i)=-2i+b$, $P(i)=c-1+bi$, and $P(-i)=c-1-bi$, \eqref{iSystem} takes the form
\begin{equation}\label{iSystem1}
\begin{split}
    cQ'(i)+b(Q(i)+iQ'(i))=-2iQ(i)+Q'(i)\\
    cQ'(-i)+b(Q(-i)-iQ'(-i))=2iQ(-i)+Q'(-i).
\end{split}
\end{equation}
Hence we can find such $b$ and $c$ if the determinant for \eqref{iSystem1} does not vanish. Let us write the determinant
\[
\begin{split}
\begin{vmatrix}
Q'(i)&Q(i)+iQ'(i)\\
Q'(-i)&Q(-i)-iQ'(-i)
\end{vmatrix}&=|Q(i)|^2
\begin{vmatrix}
\frac{Q'(i)}{Q(i)}&1+i\frac{Q'(i)}{Q(i)}\\
\frac{Q'(-i)}{Q(-i)}&1-i\frac{Q'(-i)}{Q(-i)}
\end{vmatrix}\\
&=|Q(i)|^2\left(\frac{Q'(i)}{Q(i)}-\overline{\left(\frac{Q'(i)}{Q(i)}\right)}-2i\left|\frac{Q'(i)}{Q(i)}\right|^2\right).
\end{split}
\]
Since 
\[
\frac{Q'(x)}{Q(x)}=\sum_{j=1}^{n-1}\frac{1}{x-x_j}, 
\]
we see that $\Im(Q'(i)/Q(i))<0$ and hence the determinant in question is not $0$, which proves the desired result. 

To show that $\{x_1, x_2, \dots x_{n-1}\}$ are the only real zeros of $f(x)$, just note that if $f(x)$ had all real zeros, then by the mean value theorem, $f'(x)$ must have $n$ real zeros. But this is not possible since $f'(x)$ is degree $n$ and $f'(i)=f'(-i)=0$. 
\end{proof}

\begin{proposition}\label{Rni}
For the polynomials $R_n(x)$, provided they exist, we have that $R_n(i)\neq 0$ for any $n= 1, 2, \dots$. 
\end{proposition}

\begin{proof}
Assume that $R_n(i)=0$ for some $n \geq 1$. Then since $R_n(x)$ is a real polynomial, we must also have that $R_n(-i)=0$. This combined with the fact that $R_n'(i)=R'_n(-i)=0$ from $\eqref{eq:Rn'(i)=0}$, shows that $R_n(x)/(1+x^2)^2$ is a polynomial of degree $n-2$. Let $F_{n-2}(x)=\frac{R_n(x)}{(1+x^2)^2}$  and let $\{x_1, x_2, \dots, x_k\}$, be the distinct, real roots of odd degree of $F_{n-2}(x), \quad k \leq n-2$. From Lemma \ref{lemma:zeros}, we know that there exists a polynomial $f$ of degree $k+2$ such that its only real roots are $x_j$ for $j=1,2, \dots, k$ and $f'(i)=f'(-i)=0$. Since $\{R_0(x), x, x^2, R_1(x), R_2(x), \dots, R_{k}(x)\}$ forms a basis for polynomials of degree at most $k+2$, we can write
\[f(x) = c_0R_0(x)+c_1x+c_2x^2+\sum_{j=3}^{k+2} c_jR_{j-2}(x).
\] By the condition $f'(i)=f'(-i)=0$, it must be the case that $c_1=c_2=0$. 
By the orthogonality of $\tcnew{\{R_n(x)\}_{n\geq 0}}$, we have  
\begin{align*}
    \int_{-\infty}^{\infty}  F_{n-2}(x)f(x) d\mu(x) &=\sum_{j=0, j \neq 1,2}^{k+2} c_j \int_{-\infty}^{\infty}  F_{n-2}(x)R_{j-2}(x) d\mu(x)\\
    &=\sum_{j=0, j \neq 1,2}^{k+2} c_j \int_{-\infty}^{\infty}  R_n(x)R_{j-2}(x)\frac{d\mu(x)}{(1+x^2)^2}\\
    &=0
\end{align*} However, $F_{n-2}(x)f(x)$ is a polynomial of degree $n-2+k$ where all the real roots have even multiplicity. Thus $\int_{-\infty}^{\infty} F_{n-2}(x)f(x) d\mu(x) =0$ implies that  $F_{n-2}(x)f(x)\equiv 0$ which is a contradiction. Thus, $R_n(i)\neq 0$ for all $n \geq 1$.
\end{proof}

\section{The relation to biorthogonal polynomials}

A generalization of the classical concept of orthogonality that is also relevant to our considerations was introduced by Iserles and N\o rsett \cite{IN88}. Later it was even further generalized by Brezinski \cite{B89} (see also \cite{B92}). For convenience of the reader let us recall the Brezinski setting: given a sequence of linear functionals $c^{(m)}$ defined on polynomials, find a family of monic polynomials $P_k$ such that:
\begin{enumerate}
    \item[$\bullet$] $P_k$ has the exact degree $k$;
    \item[$\bullet$] $c^{(j)}(P_k)=0$ for $j=0, 1,\dots, k-1$.
\end{enumerate}
Following Iserles and N\o rsett, we will call such polynomials $P_k$ biorthogonal provided they exist for all $k$'s. It is not so difficult to see that such polynomials can be found by the formula
\begin{equation}\label{biorthogonalpolynomials}
 P_k(x)=\frac{1}{\Delta_k}
\begin{vmatrix}
    1&x&\dots&x^k\\
    c_0^{(0)}&c_1^{(0)}&\dots&c_k^{(0)}\\
    &\hdotsfor{2}&\\
    c_0^{(k-1)}&c_1^{(k-1)}&\dots&c_k^{(k-1)}
\end{vmatrix}, \quad c_n^{(j)}=c^{(j)}(x^n)   
\end{equation}
provided that 
\begin{equation}\label{NonDegenCond}
\Delta_k=\det(c_l^{(m)})_{l,m=0}^{k-1}\ne 0.     
\end{equation}
If $\Delta_k\ne 0$ it is said that a family of the functionals $c^{(m)}$ is regular. It is worth mentioning that if for a given functional $c=c^{(0)}$ we set 
\[
c^{(m)}(p(x))=c^{(0)}(x^mp(x)),
\]
the above-described biorthogonality becomes the conventional orthogonality  with respect to the functional $c$. 

To show how this concept appears in the context of the exceptional polynomials in question, let us consider the following two functionals:
\begin{equation}\label{InitialFunct}
  c^{(0)}(p(x))=p'(i), \quad c^{(1)}(p(x))=p'(-i),  
\end{equation}
which are naturally related to \eqref{eq:Rn'(i)=0} and the fact that the polynomials $R_n$ are real. Unfortunately, we immediately see that
\[
\Delta_1=|0|=0, \quad \Delta_2=\begin{vmatrix}
    0&1\\
    0&1
\end{vmatrix}=0,
\]
which means that any family of functionals that starts with $c^{(0)}$ and $c^{(1)}$ is not regular. It also implies that we cannot construct polynomials of degrees 1 and 2 that will be biorthogonal but it is exactly what we expect when generating $R_n$. To define the rest of the family to produce $R_n$, let us introduce the functions
\begin{equation}\label{psi}
\psi_0(x)=1, \quad \psi_k(x)=x^{k+2}+\frac{k+2}{k}x^k, \quad k=1, 2,\dots.
\end{equation}
Now we are in the position to define the functionals:
\begin{equation}\label{TheRestOfFunct}
c^{(k+1)}(p(x))=\int_{-\infty}^{\infty} p(x)\psi_{k-1}(x)\frac{d\mu(x)}{(1+x^2)^2}, \quad k=1,2,\dots    
\end{equation}
and for which the following statement holds true.
\begin{theorem}
    The family of polynomials $R_n$ is {\it almost} biorthogonal with respect to the family of the functionals defined by \eqref{InitialFunct} and \eqref{TheRestOfFunct}, that is,
\[
c^{(j)}(R_k)=0, \quad j=0,\dots, k+1.
\]
Here, the term ``almost" indicates that we have to skip two degrees.
\end{theorem}
\begin{remark}
Recall that $\deg R_k=k+2$ and based on the theory of biorthogonal polynomials, the reason we have to skip the polynomials is because the family of the functionals is not regular, which is the same situation that happens for indefinite orthogonal polynomials \cite{KL79}, where the term ``almost orthogonal" was used and for Pad\'e approximation \cite{BG} (see also \cite{DD07} where the interplay between indefinite orthogonal polynomials and the Pad\'e table is given). 
\end{remark}
\begin{proof}
 The proof is immediate from Theorem \ref{thrm:Rn} and the representation
 \[
 R_k(x)=\sum_{n=0}^k\alpha_n\psi_n(x)
 \]
that holds for some coefficients $\alpha_n$, which in turn follows from the fact that $R_k'(x)=(1+x^2)\times \text{polynomial of degree k-1}$ and $\psi_n'(x)=(n+2)(1+x^2)x^{n-1}$.
\end{proof}

\begin{corollary}
    The polynomials $R_k$ can be computed by \eqref{biorthogonalpolynomials} using the moments $c_n^{(j)}$ of the linear functionals defined by \eqref{InitialFunct} and \eqref{TheRestOfFunct}.
\end{corollary}
\begin{proof}
    Using the standard \tcnew{C}ramer's rule argument, we get that the existence of the polynomial $R_k$ of degree $k+2$ implies that the corresponding determinant $\Delta_{k+2}\ne 0$ (e.g. see the proof of \cite[Theorem 1]{IN88}).  
\end{proof}
At this point we can reformulate Proposition \ref{existence} in a form that is more transparent and typical for orthogonal polynomials.
\begin{corollary}
 The polynomials $R_k$ defined as in equation \eqref{eq:Rn} which satisfy  \eqref{eq:Rn'(i)=0}, \eqref{eq:integral_R0} and \eqref{eq:integral_R1} exist if and only if
 \[
 \Delta_k\ne 0, \quad k=4,5,6,\dots,
 \]
 where $\Delta_k$ is defined by \eqref{NonDegenCond} with the functionals given by \eqref{InitialFunct} and \eqref{TheRestOfFunct}.
\end{corollary}

\begin{remark}
Note that we should have actually started with $\Delta_3$ because $\deg R_1=3$. However, we see that 
\[
\Delta_3=  \begin{vmatrix}
    0&1&2i\\
    0&1&-2i\\
    1&0&\mu_2
\end{vmatrix}=4i\ne0
\]
regardless of the value 
\[
\mu_2=\int_{-\infty}^{\infty} x^2\frac{d\mu(x)}{(1+x^2)^2}. 
\]
As a result, for the family of functionals under consideration, from \eqref{biorthogonalpolynomials} for the polynomial of degree 3 we have that 
\[
R_1(x)=\frac{1}{\Delta_3}\begin{vmatrix}
    1&x&x^2&x^3\\
    0&1&2i&-3\\
    0&1&-2i&-3\\
    1&0&\mu_2&0
\end{vmatrix}=x^3+3x
\]
\end{remark}
To conclude this section, note that there is an analog of \eqref{InitialFunct} for any family of exceptional Hermite polynomials and thus the results of this section can be adapted to the general case and thus exceptional Hermite polynomials are a subclass/limiting case of a larger class of biorthogonal polynomials that has various applications and generic properties.

\begin{section}{Modification of the Christoffel Formula}
By definition, given a family of orthogonal polynomials $\tcnew{\{P_n(x)\}_{n\geq 0}}$, we can obtain the family of exceptional orthogonal polynomials $\tcnew{\{R_n(x)\}_{n\geq 0}}$. We aim to reverse the process and obtain the polynomials $P_n(x)$ from the polynomials $R_n(x)$. 

Since by Theorem \ref{thrm:Rn}, the polynomials $R_n(x)$ are orthogonal with respect to $\frac{d\mu(x)}{(1+x^2)^2}$, one would expect to get the monic polynomials $P_n(x)$, under the classical Christoffel transformation of $R_n(x)$. Thus, applying  \cite[Theorem 2.7.1]{I09} and letting $\phi(x)=(1+x^2)^2= (x-i)^2(x+i)^2$, one hopes that  $P_n(x)$ can be defined by
\begin{equation}\label{eq:christoffel}
    C_{n}\phi(x)P_n(x) = \begin{vmatrix}
    R_n(i)&R_{n+1}(i)&R_{n+2}(i)&R_{n+3}(i)&R_{n+4}(i)\\
     R'_n(i)&R'_{n+1}(i)&R'_{n+2}(i)&R'_{n+3}(i)&R'_{n+4}(i)\\
     R_n(-i)&R_{n+1}(-i)&R_{n+2}(-i)&R_{n+3}(-i)&R_{n+4}(-i)\\
     R'_n(-i)&R'_{n+1}(-i)&R'_{n+2}(-i)&R'_{n+3}(-i)&R'_{n+4}(-i)\\
     R_n(x)&R_{n+1}(x)&R_{n+2}(x)&R_{n+3}(x)&R_{n+4}(x)\\
     \end{vmatrix}
\end{equation} where 
\begin{equation*}
    C_n =\begin{vmatrix}
    R_n(i) & R_{n+1}(i)&R_{n+2}(i)&R_{n+3}(i)\\
    R'_n(i) & R'_{n+1}(i)&R'_{n+2}(i)&R'_{n+3}(i)\\
    R_n(-i) & R_{n+1}(-i)&R_{n+2}(-i)&R_{n+3}(-i)\\
    R'_n(-i) & R'_{n+1}(-i)&R'_{n+2}(-i)&R'_{n+3}(-i)
    \end{vmatrix}.
\end{equation*}
However, by \eqref{eq:Rn'(i)=0}, we have that $R'_n(i)=R'_n(-i)=0 \text{ for all }n=0,1,2, \dots$, therefore the determinant on the right-hand side of \eqref{eq:christoffel} vanishes as well as $C_n=0$ and
we cannot make any conclusions about the polynomials in \eqref{eq:christoffel}.
To resolve this issue, we instead define a sequence of polynomials $\tcnew{\{S_n(x)\}_{n\geq 0}}$ as follows: 
\begin{equation}\label{def:sn}S_n(x):= \frac{1}{c_n\phi(x)}\begin{vmatrix}R_n(i)&R_{n+1}(i)&R_{n+2}(i)\\
R_n(-i)&R_{n+1}(-i)&R_{n+2}(-i)\\
R_n(x)&R_{n+1}(x)&R_{n+2}(x)
\end{vmatrix}
\end{equation} where 
\begin{equation}\label{def:cn}c_n=\begin{vmatrix}
R_n(i)&R_{n+1}(i)\\
R_n(-i) & R_{n+1}(-i)
\end{vmatrix}.
\end{equation}
Note that $c_n\neq 0$ for any $n \geq 0$ since \[R_n(i)R_{n+1}(-i)-R_{n+1}(i)R_n(-i) = \pm 2 R_n(i)R_{n+1}(i)\] and by Proposition \ref{Rni}, $R_n(i)\neq 0$ for any $n \geq 0$.

In general we have the following:

\begin{proposition}  $S_n(x)$ is a real, monic polynomial of degree $n$ and \[ S_n(-x)=(-1)^nS_n(x).\]
\end{proposition}
\begin{proof} Let \[D_n(x)=\begin{vmatrix}R_n(i)&R_{n+1}(i)&R_{n+2}(i)\\
R_n(-i)&R_{n+1}(-i)&R_{n+2}(-i)\\
R_n(x)&R_{n+1}(x)&R_{n+2}(x)
\end{vmatrix}.\]  Then $D_n(x)$ has a zero at $i$ and $-i$ since a row will be repeated. Also, $D_n'(x)=a_nR_n'(x)+b_nR'_{n+1}(x)+c_nR'_{n+2}(x)$ for complex numbers $a_n, b_n, c_n$. Since $R'_n(x)$ has zeros at $i$ and $-i$ for all $n \geq 1$, $D_n(x)$ must have zeros of multiplicity $k \geq 2$ at $i$ and $-i$, therefore, $D_n(x)/\phi(x)$ is a polynomial of degree $\leq n$. Since the leading coefficient of $D_n(x)$ is $c_n$ from \eqref{def:cn}, which is nonzero, $D_n(x)/\phi(x)$ has degree n. \\
\indent Also notice that $b_n=R_n(i)R_{n+2}(-i)-R_{n}(-i)R_{n+2}(i)=0$ for all $n\geq 0$ since
if $n$ is even, then $R_n(x)$ is an even function so $R_n(i)=R_n(-i)$ and thus $R_n(i)R_{n+2}(-i)=R_n(-i)R_{n+2}(i)$. Similarly, if $n$ is odd, $R_n(x)$ is an odd function, so $R_n(i)=-R_n(-i)$. Therefore, $
R_n(i)R_{n+2}(-i)=(-R_n(-i))(-R_{n+2}(i))=R_n(-i)R_{n+2}(i)$. In both cases we see that $b_n=0$ for all $n \geq 0$. Hence, the assertion that $S_n(-x)=(-1)^nS_n(x)$ follows simply from the fact that \begin{equation}\label{eq:sn}\phi(x)S_n(x)=\frac{a_n}{c_n}R_n(x)+R_{n+2}(x) 
\end{equation} and $R_n(-x)=(-1)^nR_n(x)$. Lastly, equation \eqref{eq:sn} shows that $\overline{S_n(x)}=S_n(x)$ since the $R_n(x)'s$ are real polynomials, thus $S_n(x)$ must have real coefficients for all $n=0,1,2\dots.$ 
\end{proof}
It is natural to ask about orthogonality relations regarding the $S_n(x)$ and in fact, we have that they are ``almost" bi-orthogonal to the polynomials $R_n(x)$.
\begin{theorem}\label{thm:Sn_Rn} 
For the sequences $\tcnew{\{S_n(x)\}_{n\geq 0}}$ and $\tcnew{\{R_n(x)\}_{n\geq 0}}$ we have that
\[
\int_{-\infty}^\infty S_n(x)R_m(x) d\mu(x)=0 \text{ for } m=0,1,\dots, n-1, n+1, n+3, n+4,\dots \] and 
\[
\int_{-\infty}^\infty S_n(x)R_m(x) d\mu(x) \neq 0 \text{ for } m=n, n+2 \]
\end{theorem}
\begin{proof} For the cases $m < n$, $m=n+1$, or $m>n+2$, we see
 by the orthogonality of $\{R_n(x)\}$, that
\begin{align*}\int_{-\infty}^\infty S_n(x)R_m(x) d\mu(x)&=
\int_{-\infty}^\infty \frac{a_n}{c_n}R_n(x)R_m(x)\frac{d\mu(x)}{(1+x^2)^2} \\
&+\int_{-\infty}^\infty R_{n+2}(x)R_m(x)\frac{d\mu(x)}{(1+x^2)^2} \\
&=0.
\end{align*} 
Lastly, we have
\[
\int_{-\infty}^\infty S_n(x)R_m(x) d\mu(x)\neq 0 \text{ for } m=n, n+2 \] since $a_n=c_{n+1}\neq 0$ for all $n=0,1,2,\dots$ by \eqref{def:cn}.
\end{proof}
It turns out that one can state the above theorem for polynomials $f(x)$ satisfying $f'(i)=f'(-i)=0$.
\begin{theorem}\label{thm:sn_f}
Let $f(x)\in \mathbb{C}[x]$ such that $f'(i)=f'(-i)=0$ and let deg($f)=m$. Then if $m<n+2$,
\[\int_{-\infty}^\infty S_n(x)f(x) d\mu(x) =0
\] 
\end{theorem}
\begin{proof}
Since $R_k'(i)=R_k'(i)=0$ for all $k=0,1,2,\dots$, $\frac{R'_k(x)}{1+x^2}$ is a polynomial of degree $k-1$ for $k=1,2,\dots$ so $\left\{\frac{R'_k(x)}{1+x^2}\right\}_{k=1}^\infty$ is a basis for $\mathbb{C}[x]$. Thus, $\frac{f'(x)}{1+x^2}=\sum_{k=1}^{m-2}a_k\frac{R'_k(x)}{1+x^2}$. Multiplying by $1+x^2$ and then integrating we see 
\[
f(x)=\sum_{k=0}^{m-2}a_k R_k(x).
\]
The result now follows from Theorem \ref{thm:Sn_Rn}.
\end{proof}
The biorthogonality relationship from Theorem \ref{thm:Sn_Rn} allows us to define new polynomials which will be shown to coincide with our original polynomials $P_n(x)$.

\begin{theorem}\label{thm:Rho_n}  Let $\rho_n \in \mathbb{R}$ be such that
\begin{equation}\label{def:rho_n}\rho_n =
\begin{cases}
0 & n= 1,2\\
\frac{-\int_{-\infty}^\infty xS_n(x)d\mu(x) }{\int_{-\infty}^\infty xS_{n-2}(x)d\mu(x) } & n \text{ odd, } n\geq 3\\
\frac{-\int_{-\infty}^\infty x^2S_n(x)d\mu(x)  }{\int_{-\infty}^\infty x^2S_{n-2}(x)d\mu(x) } & n \text{ even, } n\geq 4.
\end{cases}
\end{equation}
Then defining $S_{-1}(x):=0$, the monic polynomial $\mathcal{P}_n(x):= S_n(x)+\rho_nS_{n-2}(x)$ is orthogonal to $\{1, x, x^2, R_1(x), R_2(x), \dots, R_{n-3}(x)\}$  with respect to $\mu$ for all $n \geq 1$.
\end{theorem}
\begin{proof} 
To show $\rho_n$ is well defined for $n\geq 3$, consider the case where $n$ is odd and assume by way of contradiction that $\int_{-\infty}^\infty xS_{n-2}(x)d\mu(x)=0 $. We claim this implies that if $k$ is a positive, odd integer, then
\begin{equation}\label{eq:x^k_Sn}\int_{-\infty}^\infty x^kS_{n-2}(x)d\mu(x) =0
\end{equation} for all $n>k$. Recall that $R_{k}(x)$ is a monic polynomial of degree $k+2$ and by Theorem \ref{thm:Sn_Rn},
\[\int_{-\infty}^\infty S_{n-2}(x)R_k(x)d\mu(x)=0 \] for $n>k+2$. But, 
\begin{multline*}
\int_{-\infty}^\infty S_{n-2}(x)R_k(x)d\mu(x) = \int_{-\infty}^\infty x^{k+2}S_{n-2}(x)d\mu(x)\\
+\int_{-\infty}^\infty x^{k}S_{n-2}(x)d\mu(x)+\dots+ \int_{-\infty}^\infty xS_{n-2}(x)d\mu(x)
\end{multline*} thus, by induction, $\int_{-\infty}^\infty x^kS_{n-2}(x)d\mu(x) =0$ for all $n>k$. Therefore, equation $\eqref{eq:x^k_Sn}$ implies that $S_n(x)$ is orthogonal to $\{x, x^3, x^5, \dots, x^n\}$ for all $n$ odd, $n \geq 1$. In particular,
\[\int_{-\infty}^\infty S_{n}^2(x)d\mu(x)=0
\] which is a contradiction. Therefore, $\rho_n$ is well defined for all $n$ odd, $n \geq 1$. \\
\indent The same reasoning holds for when $n$ is even since then we will have that $S_n(x)$ is orthogonal to $\{1, x^2, x^4, \dots, x^n\}$ for all $n \geq 2$. Therefore, $\rho_n$ is well-defined for all $n \geq 1$.\\
\indent It remains to show the orthogonality of $\mathcal{P}_n(x)$ with each polynomial in \newline $\{1, x, x^2, R_1(x), R_2(x), \dots, R_{n-3}(x)\}$. Notice that
\begin{align*}
    \int_{-\infty}^\infty \mathcal{P}_n(x)R_k(x) d\mu(x) &= \int_{-\infty}^\infty S_n(x)R_k(x)d\mu(x) +\rho_n\int_{-\infty}^\infty S_{n-2}(x)R_k(x)d\mu(x)  \\
    &=0 \text{ if }k<n-2
\end{align*} by Theorem \ref{thm:Sn_Rn}. Thus, $\int_{-\infty}^\infty \mathcal{P}_n(x)R_k(x)d\mu(x)=0 $ for all $k=1, 2, \dots, n-3$.\\
\indent Now, 
\begin{align*}
    \int_{-\infty}^\infty \mathcal{P}_n(x) d\mu(x)&= \int_{-\infty}^\infty S_n(x)d\mu(x)+\rho_n\int_{-\infty}^\infty S_{n-2}(x) d\mu(x)\\
    &= \int_{-\infty}^\infty \left(\frac{a_n}{c_n}R_n(x)+R_{n+2}(x)\right) \frac{d\mu(x)}{(1+x^2)^2} dx\\
    &+ \rho_n\int_{-\infty}^\infty\left(\frac{a_{n-2}}{c_{n-2}}R_{n-2}(x)+R_n(x)\right)\frac{d\mu(x)}{(1+x^2)^2} dx\\
    &=0 \text{ for } n\geq 3
\end{align*} where the last equality holds by the orthogonality of $R_n(x)$ with $R_0(x)$. If $n=1$ then 
\[\int_{-\infty}^\infty \mathcal{P}_1(x)d\mu(x)= \int_{-\infty}^\infty S_1(x) d\mu(x)=0\] since $S_1(x)$ is odd (and by the orthogonality of $R_1(x)$ and $R_3(x)$ with $R_0(x)$).\\
\indent If $n=2$, then
\[\int_{-\infty}^\infty \mathcal{P}_2(x)d\mu(x)= \int_{-\infty}^\infty S_2(x) d\mu(x)+\rho_2\int_{-\infty}^\infty S_{0}(x)d\mu(x)
=0\] by definition of $\rho_2$ and $S_2(x)$, so that $\mathcal{P}_2(x)$ is orthogonal to 1 and thus $\mathcal{P}_n(x)$ is orthogonal to 1 for all $n \geq 1$.\\
\indent It is now easy to see that $\mathcal{P}_n(x)$ is orthogonal to $x$ for all $n \geq 2$ since if $n$ is even, $x\mathcal{P}_n(x)$ is an odd function so that
\[\int_{-\infty}^\infty x\mathcal{P}_n(x)d\mu(x) =0.
\]  If $n$ is odd, then
\begin{equation*}\int_{-\infty}^\infty x\mathcal{P}_n(x)d\mu(x) = \int_{-\infty}^\infty xS_n(x)d\mu(x)+\rho_n\int_{-\infty}^\infty xS_{n-2}(x)d\mu(x) =0
\end{equation*}
by the definition of $\rho_n$. Thus, $\mathcal{P}_n(x)$ is orthogonal to $x$ for $n \geq 2$. \\
\indent Similarly, $\mathcal{P}_n(x)$ is orthogonal to $x^2$ for $n\geq 3$ since if $n$ is odd, then 
\[\int_{-\infty}^\infty x^2\mathcal{P}_n(x)d\mu(x) =0
\] since $x^2\mathcal{P}_n(x)$ is an odd function, and if $n$ is even then the result follows by definition of $\rho_n$.\\
\indent Thus, $\mathcal{P}_n(x)$ is orthogonal to $\{1, x, x^2, R_1(x), R_2(x), \dots ,R_{n-3}(x)\}$ with respect to 
$\mu$ as wanted.
\end{proof}

\begin{theorem}\label{GenChristoffel} Let $\tcnew{\{P_n(x)\}_{n\geq 0}}$ be a sequence of monic orthogonal polynomials with respect to a symmetric measure $\mu$  let $\tcnew{\{R_n(x)\}_{n\geq 0}}$ be the family of polynomials defined by $R_n(x)=P_{n+2}(x)+A_nP_n(x)+B_n(x)P_{n-2}(x)$. Assume that $\tcnew{\{R_n(x)\}_{n\geq 0}}$ satisfies the following:
\begin{enumerate}
    \item $R'_n(i)=0 \quad n=1,2,\dots$
\item $\int_{-\infty}^{\infty}  R_0(x)R_n(x)\frac{d\mu(x)}{(1+x^2)^2}=0 \quad n=1,2,\dots$
     \item $\
     \int_{-\infty}^{\infty}  R_1(x)R_n(x)\frac{d\mu(x)}{(1+x^2)^2}=0 \quad n \geq 2$.
\end{enumerate}
Then 
\begin{multline*}
    \phi(x)P_n(x)= \\
= \left[ \frac{1}{c_n}\begin{vmatrix}R_n(i)&R_{n+1}(i)&R_{n+2}(i)\\
R_n(-i)&R_{n+1}(-i)&R_{n+2}(-i)\\
R_n(x)&R_{n+1}(x)&R_{n+2}(x)
\end{vmatrix}+\frac{\rho_n}{c_{n-2}}\begin{vmatrix}R_{n-2}(i)&R_{n-1}(i)&R_{n}(i)\\
R_{n-2}(-i)&R_{n-1}(-i)&R_{n}(-i)\\
R_{n-2}(x)&R_{n-1}(x)&R_{n}(x)
\end{vmatrix}\right]
\end{multline*}
where $\phi(x)=(1+x^2)^2$ and $\rho_n$ and $c_n$ are given by \eqref{def:rho_n} and \eqref{def:cn}, respectively. 
\end{theorem}
\begin{proof}First, since $\{1, x, x^2, R_1(x), R_2(x), \dots R_{n-3}(x)\}$ is a basis for polynomials of degree less than $n$,  we can write $x^k=a_0+a_1x+a_2x^2+a_3R_1(x)+\dots +a_kR_{k-2}(x)$ for $k < n$. Therefore, by Theorem \ref{thm:Rho_n}, $\{\mathcal{P}_n(x)\}_{n\geq 0}$ is orthogonal to $\{1, x, x^2,\dots, x^{n-1}\}$. Also, since $\mathcal{P}_n(x)$ is a polynomial of degree $n$, this orthogonality implies that $\int_{-\infty}^\infty x^n\mathcal{P}_n(x)d\mu(x)\neq 0$ for all $n \geq 1$ (otherwise $\mathcal{P}_n(x) \equiv 0$ for all $n \geq 1$). Therefore, $\tcnew{\{\mathcal{P}_n(x)\}_{n\geq 0}}$ is a monic OPS with respect to $\mu$. By the uniqueness of orthogonal polynomial systems, we must have that $P_n(x)=\mathcal{P}_n(x)$.
\end{proof}

\tcnew{It is interesting to note that while the polynomials $S_n(x)$ defined in \eqref{def:sn} may not form an orthogonal polynomial sequence (as will be shown in Section 6), they still have familiar behavior of their zeros.}
\begin{proposition}\label{prop:Sn_Zeros}
The zeros of $S_n(x)$ are real and simple \tcnew{and lie in the interior of the support of $d\mu(x)$}.
\end{proposition}

\begin{proof}
First, since $S_n(x) \neq 0$ for $n=0,1,2,\dots$ and \[\int_{-\infty}^\infty S_n(x) d\mu(x) =\int_{-\infty}^\infty S_n(x)R_0(x) d\mu(x)=0\] for all $n=1,2,\dots$, it must be the case that \tcnew{for for $n=1,2,\dots$,  $S_n(x)$ has at least one zero of odd multiplicity which lies in the interior of the support of $d\mu(x)$.} So let $x_1, x_2, \dots x_k$, $1\leq k \leq n$, be the distinct zeros of odd multiplicity of $S_n(x)$ \tcnew{in the interior of the support of $d\mu(x)$}. By Lemma $\ref{lemma:zeros}$, there exists a polynomial $f(x)$ of degree $k+2$ such that $f(x_j)=0$ for $j=1, 2,\dots, k$ and $f(i)=f(-i)=0$. Then $S_n(x)f(x)e^{-\frac{x^2}{2}}\geq 0$ for all $x \in(-\infty, \infty)$. But, by Theorem $\ref{thm:sn_f}$,
\[\int_{-\infty}^\infty
S_n(x)f(x) d\mu(x) =0\]
for $k<n$ so we must have that $k=n$ and hence $S_n(x)$ has $n$ distinct, real zeros \tcnew{in the interior of the support of $d\mu(x)$.} 
\end{proof} 
\tcnew{We also have that the zeros of $S_{n}(x)$ and $S_{n+1}(x)$ interlace. Before proving this behavior, we recall the definition of an integrable Markov system (see \cite{K70}).}
\begin{definition}
   \tcnew{ Let $m_k(x), k=0,1,2,\dots$, be real valued functions on $\mathbb{R}$. Then the sequence $\{m_k(x)\}_{k\geq 0}$ forms an \textbf{integrable Markov system} on $(a,b)$ if
    \begin{enumerate}
        \item For each $k=0,1,2,\dots, w_k(x)$ is defined on $(a,b)$ and $\int_{a}^{b}m_k(x) \, dx<\infty$ .
        \item For $n=1,2,\dots$ and arbitrary scalars $a_0, a_1,\dots$ (not all zero), the function
        \[
        f(x):=\sum_{k=0}^{n-1} a_km_k(x)
        \] has at most $n-1$ zeros in $(a,b)$. 
    \end{enumerate}}
\end{definition}
We are now in a position to show the interlacing property of the zeros of $S_n(x)$ and $S_{n+1}(x)$ when the associated measure has density with respect to the Lebesgue measure. 
\begin{proposition}\label{thm:Sn_Zeros}
    \tcnew{Let $\{P_n(x)\}$ be a sequence of polynomials orthogonal with respect to the positive weight function $w(x)$ on $(a,b)$ and let $\{R_n(x)\}_{n\geq 0}$ and $\{S_n(x)\}_{n\geq 0}$ be the corresponding families of polynomials defined in \eqref{eq:Rn} and \eqref{def:sn}, respectively.  
 Denote the $k$-th zero of $S_n(x)$ by ${x_{nk}}$. Then 
    \[
    x_{n+1,k}<x_{nk}<x_{n+1,k+1}, 
    \quad k=1,2, \dots, n.
    \]}
\end{proposition}
\begin{proof}
\tcnew{Consider the family of functions $\{\psi_k(x)\}_{k\geq 0}$ where 
\[
\psi_0(x)=1 \text{ and } \psi_k(x)=x^{k+2}+\frac{k+2}{k}x^k, \quad k=1,2,\dots 
\]
as in \eqref{psi} and let $m_k(x)=w(x)\psi_k(x)$ for $k=0,1,\dots$. Then $m_k(x)$ is integrable for $k=0,1,2,\dots$ since the measure $w(x)dx$ has finite moments. Also, $\sum_{k=0}^{n-1} a_k m_k(x)=w(x)\sum_{k=0}^{n-1} a_k\psi_k(x)$ where $f(x):=\sum_{k=0}^{n-1} a_k\psi_k(x)$ is a polynomial of degree at most $n+1$. Note that $f'(i)=f'(-i)=0$, therefore $f'(x)$ has at most $n-2$ real zeros. Thus, by the mean value theorem, $f(x)$ can have at most $n-1$ real zeros. Since $w(x)$ is non-zero on $(a,b)$, we see that $\sum_{k=0}^{n-1} a_k w_k(x)$ can have at most $n-1$ zeros in $(a,b)$, hence, 
 $\{m_k(x)\}_{k\geq 0}$ forms an integrable Markov system.\\
\indent Recall that
\[
\int_{a}^{b}\psi_k(x)S_{n+1}(x)w(x)dx=0
\] for $k<n+1$ by Theorem \ref{thm:sn_f}, and the zeros of $S_{n+1}(x)$ are real and distinct by Proposition \ref{prop:Sn_Zeros}, thus the result follows from Theorem 3$(iii)$ of \cite{K70}}.

\end{proof}
\begin{remark}
  \tcnew{One can check that Theorem 3$(iii)$ of \cite{K70} can be extended to any measure supported on $\mathbb{R}$ with finite moments so that Proposition \ref{thm:Sn_Zeros} also holds in the more general case.} 
\end{remark}
\end{section}

\section{\tcnew{Recurrence relations}}

In this section we establish that the polynomials $R_n(x)$ satisfy a higher-order recurrence relation \tc{and at the same time we show that \tcnew{the $R_n$} are a discrete Darboux transformation of the original family $P_n$}. \tcnew{Note that in \cite{GKKM} it was shown that exceptional Hermite polynomials satisfy a family of recurrence relations and we demonstrate that a similar statement is also valid for the $R_n$.}

It has been shown in \cite{BS95}, \cite{SO95A}, \cite{SO95B} that the differential operator underlying \eqref{DEforF} can be obtained via a double commutation method (aka \tcnew{continuous }Darboux transformation) from the analogous differential operator corresponding to Hermite polynomials. \tc{The ideology of generating new nonclassical orthogonal polynomials from the classical ones goes back the works of \tcnew{Gr\"unbaum} and Haine (see \cite{GH97}, \cite{GHH99} and the references therein) and now} we are in the position to demonstrate that a similar situation occurs for the difference operators underlying the \tcnew{DEK-type orthogonal} polynomials and the corresponding conventional orthogonal polynomials. To this end, let us consider the monic Jacobi matrix
\begin{equation}\label{J_h}
  J=\begin{pmatrix}
0 & 1 & 0 & \cdots   \\
     a_1 & 0 & 1 &   \\
     0 & a_2 & 0 & \ddots  \\
     \vdots & & \ddots & \ddots
\end{pmatrix}  
\end{equation}
that corresponds to the family of symmetric orthogonal polynomials $\{P_n(x)\}$ that we started with, that is,
\[
J{\bf P}(x)=x{\bf P}(x),
\]
where ${\bf P}(x)=(P_0(x), P_1(x), P_2(x),\dots)^{\top}$. From \eqref{eq:Rn} and Theorem \ref{GenChristoffel} we conclude that
\begin{equation}\label{dDarboux}
 {\bf R}(x)=A{\bf P}(x), \quad \phi(x){\bf P}(x)=B{\bf R}(x),   
\end{equation}
where ${\bf R}(x)=(R_0(x), R_1(x), R_2(x),\dots)^{\top}$, $A$ and $B$ are some banded matrices. Then the following result holds true.

\begin{theorem}\label{DarbouxForR}
We have that
\begin{equation}\label{RecForF}
    AB{\bf R}(x)=\phi(x){\bf R}(x)
\end{equation}
and
\begin{equation}\label{mGeronimus}
    BA=(J^2+I)^2
\end{equation}
where $J$ is given by \eqref{J_h} and $I$ is the identity matrix.
\end{theorem}

\begin{proof}
To get \eqref{RecForF}, one multiplies the second relation in \eqref{dDarboux} by $A$ on the left and uses the first relation to get rid of ${\bf P}(x)$. Similarly, we get 
\[
BA{\bf P}(x)=\phi(x){\bf P}(x)
\]
Since $\phi(x)=(1+x^2)^2$, we have $\phi(x){\bf P}(x)=(J^2+I)^2{\bf P}(x)$ and hence
\[
BA{\bf P}(x)=(J^2+I)^2{\bf P}(x).
\]
The latter relation gives \eqref{mGeronimus} because any finite number of the orthogonal polynomials form a linearly independent system.
\end{proof}
\begin{remark}
Formula \eqref{RecForF} constitutes a recurrence relation for $R_n(x)$ and therefore \tc{the} exceptional orthogonal polynomials $R_n(x)$ satisfy a higher-order recurrence relation \tc{and form generalized eigenvectors of the corresponding non-selfadjoint operator, which allows to do spectral analysis of the underlying semi-infinite band matrices. Another message is that the approach can be applied to other similar classes of biorthogonal polynomials.  Note that this approach was already implemented for some nonclassical orthogonalities.} For instance, a discrete Darboux transformation can lead to Sobolev orthogonal polynomials \cite{DM14}, \cite{DGM14} as well as to indefinite orthogonal polynomials \cite{DD10}. As for the exceptional part, skipping polynomials of certain degrees is natural for discrete Darboux transformations as can be seen on indefinite orthogonal polynomials \cite{DD04}, \cite{DD07}, \cite{KL79}.
\end{remark}

\tcnew{A different technique leads to a family of recurrence relations analogous to what was obtained in \cite{GKKM} but it does not immediately reveal the bond between the spectral properties unlike the commutation relation given in Theorem \ref{DarbouxForR}.
\begin{proposition}\label{Mrecrel}
    Let $\psi$ be a monic polynomial of degree $k$ such that
    \[
    \psi'(i)=0,\quad \psi'(-i)=0.
    \]
 Then 
 \[
 \psi(x)R_n(x)=\sum_{m=n-k}^{n+k}c_{n,m}R_m(x)
 \]
 for some constants $c_{n,m}$.
\end{proposition}
\begin{proof}
Since $f(x)=\psi(x)R_n(x)$ satisfies $f'(i)=f'(-i)=0$, using the reasoning given in the proof of Theorem \ref{thm:sn_f}, we see that 
\[
\psi(x)R_n(x)=\sum_{m=0}^{n+k}c_{n,m}R_m(x).
\]
Reiterating the argument for $\psi(x)R_m(x)$ and using the orthogonality of $R_n$, we see that
\[
\int_{-\infty}^{\infty}\psi(x)R_n(x)R_m(x)\frac{d\mu(x)}{(1+x^2)^2}=\int_{-\infty}^{\infty}R_n(x)\left(\sum_{l=0}^{m+k}c_{m,l}R_l(x)\right)\frac{d\mu(x)}{(1+x^2)^2}=0
\]
provided that $m<n-k$, which yields the desired relation.
\end{proof}
In particular, if $\psi(x)=1+x+\frac{x^3}{3}$ then Proposition \ref{Mrecrel} gives a $7$-term recurrence relation while Proposition \ref{DarbouxForR} produces $9$-term recurrence relation for $\psi(x)=(1+x^2)^2$.}

\begin{section}{Examples}

Here we firstly show that the DEK polynomials fit into the general scheme that we presented and then apply the construction to the Chebyshev polynomials, which leads to a new family of DEK-type orthogonal polynomials that does not coincide with the known families.

\begin{subsection}{DEK Polynomials}
As was pointed out at the beginning the DEK polynomials satisfy the relations \cite{D92}
\begin{align}\label{eq:Fn}F_n(x)&=(x^3+3x)He_{n-1}(x)-(n-1)(1+x^2)He_{n-2}(x)\\
&=He_{n+2}(x)+2(n+2)He_n(x)+(n+2)(n-1)He_{n-2}(x).
\end{align} 
Next, note that from equation \eqref{eq:Fn}, one has that 
\[
F'(i)=0
\]
for all $n$. In this case, we know the coefficients from the start but, in principle, Proposition \ref{existence} could independently establish existence of the coefficients and the procedure given in the proof would lead to $A_n=2(n+2)$ and $B_n=(n+2)(n-1)$. Then, the orthogonality would follow from Theorem \ref{thrm:Rn}. Since $F_n(i)=F_n'(-i)=0$, as was already mentioned, the classic Christoffel transformation cannot be applied. However, Theorem \ref{GenChristoffel} allows us to obtain the original $He_n(x)$ from the $F_n(x)$. We show this below for $n=1,2,\dots, 5$.
\begin{example} Let $S_n(x)$ be the polynomials defined by
\begin{equation*}S_n(x):= \frac{1}{c_n(1+x^2)^2}\begin{vmatrix}F_n(i)&F_{n+1}(i)&F_{n+2}(i)\\
F_n(-i)&F_{n+1}(-i)&F_{n+2}(-i)\\
F_n(x)&F_{n+1}(x)&F_{n+2}(x)
\end{vmatrix}
\end{equation*} where
\begin{equation*}c_n=\begin{vmatrix}
F_n(i)&F_{n+1}(i)\\
F_n(-i) & F_{n+1}(-i)
\end{vmatrix}.
\end{equation*} 

Then applying Theorem \ref{GenChristoffel} to the DEK polynomials for $n=0,1,2,3,4$, we have 
\begin{align*}
\phi(x)He_0(x)&=\phi(x)S_0(x)=F_2(x)+2F_0(x)\\
\phi(x)He_1(x)&=\phi(x)S_1(x)=F_3(x)+2F_1(x)\\
\phi(x)He_2(x)&=\phi(x)S_2(x)=F_4(x)+4F_2(x)\\
\phi(x)He_3(x)&=\phi(x)[S_3(x)+S_1(x)]=F_5(x)+6F_3(x)+2F_1(x)\\
\phi(x)He_4(x)&=\phi(x)[S_4(x)+\frac{3}{2}S_2(x)]= F_6(x)+8F_4(x)+6F_2(x)
\end{align*}
\end{example}
We list the first five $S_n(x)$ for the reader's convenience. 
\begin{align*}
S_0(x)&=1\\
    S_1(x)&=x\\
    S_2(x)&=x^2-1\\
    S_3(x)&=x^3-4x\\
    S_4(x)&=x^4-\frac{15}{2}x^2+\frac{9}{2}
\end{align*} 

One may ask if $\{S_n(x)\}$ is an orthogonal polynomials system.
If this were the case then 
$\{S_n(x)\}$ would satisfy the three-term recurrence relation
\[xS_n(x)=S_{n+1}(x) + \alpha_nS_n(x)+\beta_nS_{n-1}(x).
\] 
However, if we consider when $n=3$, then we cannot find coefficients $\alpha_3$ and $\beta_3$ such that this three-term recurrence relation is satisfied. Thus, the polynomials $S_n(x)$ cannot be an OPS with respect to any quasi-definite linear functional.
\begin{remark}
It was shown in \cite{A} that the DEK polynomials are complete in \newline $L^2\left(\frac{e^{-x^2/2}}{(1+x^2)^2}dx, (-\infty, \infty) \right).$
\end{remark}
\end{subsection}

\begin{subsection}{Chebyshev Polynomials}
We now take the case $P_n(x)=\hat{T}_n(x)$ where $\hat{T}_n(x)$ is the monic Chebyshev polynomial \tcnew{of the first kind} of degree $n$ and $d\mu(x) = (1-x^2)^{-1/2}dx$.
\begin{theorem}\label{Thm:ExCheb} Let $\hat{T}_n(x)$ denote the monic Chebyshev polynomial \tcnew{of the first kind} of degree $n$.  Then for $n \geq 1$, there exist real numbers $A_n$ and  $B_n$ such that $\mathcal{R}_n(x)=\hat{T}_{n+2}(x)+A_n\hat{T}_n(x)+B_n\hat{T}_{n-2}(x)$ is a monic polynomial of degree $n+2$ satisfying\tcnew{,}
\begin{enumerate}
    \item $\mathcal{R}'_n(i)=0$ for all $n=1,2,\dots$\tcnew{,}
    \item $\int_{-1}^1  \frac{\mathcal{R}_0(x)\mathcal{R}_n(x)}{\sqrt{1-x^2}(1+x^2)^2}dx=0$ for all $n=1,2,\dots$\tcnew{,}
    \item $\int_{-1}^1  \frac{\mathcal{R}_1(x)\mathcal{R}_n(x)}{\sqrt{1-x^2}(1+x^2)^2}dx=0$ for all $n \geq 2$
\end{enumerate}
where $\mathcal{R}_0(x):=1$.
\end{theorem}
\begin{proof} 

First, it should be noted that if $A_n, B_n \in \mathbb{R}$ and $\mathcal{R}_n'(i)=0$ then, it must be the case that $\mathcal{R}_n'(-i)=0$. 

We can explicitly find $\mathcal{R}_1(x)$ since $\mathcal{R}_1(x)=
\hat{T}_3(x)+A_1\hat{T}_1(x)$, is an odd, monic polynomial of degree 3 which satisfies that $\mathcal{R}'_1(i)=0$. Thus $\mathcal{R}_1(x)=x^3+3x$ and hence $A_1=15/4$ and $B_1$ can be arbitrary. Notice that $\mathcal{R}_1(x)$ satisfies $\int_{-1}^1 \frac{\mathcal{R}_0(x)\mathcal{R}_1(x)}{\sqrt{1-x^2}(1+x^2)^2}dx=0$ since $\frac{\mathcal{R}_0(x)\mathcal{R}_1(x)}{\sqrt{1-x^2}(1+x^2)^2}$ is an odd function.

Now for $n \geq 2$, by Proposition \ref{existence}, it suffices to show that that for for $n \geq 2$ even, det$(\mathcal{E}_n)\neq 0$ and for $n \geq 3$ odd, det$(\mathcal{O}_n) \neq 0$
where \[\mathcal{E}_n = \begin{pmatrix}
\hat{T}'_n(i)&\hat{T}'_{n-2}(i)\\
\int_{-1}^1 \frac{\hat{T}_n(x)}{\sqrt{1-x^2}(1+x^2)^2}dx  & \int_{-1}^1 \frac{\hat{T}_{n-2}(x)}{\sqrt{1-x^2}(1+x^2)^2}dx
\end{pmatrix}\] and  \[\mathcal{O}_n= \begin{pmatrix}
\hat{T}'_n(i)&\hat{T}'_{n-2}(i)\\
\int_{-1}^1 \frac{\mathcal{R}_1(x)\hat{T}_n(x)}{{\sqrt{1-x^2}(1+x^2)^2}} dx & \int_{-1}^1 \frac{\mathcal{R}_1(x)\hat{T}_{n-2}(x)}{{\sqrt{1-x^2}(1+x^2)^2}} dx
\end{pmatrix}.\] In fact, since $\hat{T}_n(x) = \frac{1}{2^{n-1}}T_n(x)$ for $n \geq 2$, where $T_n(x)$ is non-monic Chebyshev polynomial of the first kind of degree $n$, we can replace $\hat{T}$ with $T$ in $\mathcal{E}_n$ and $\mathcal{O}_n$ and show the corresponding determinants are non-zero. Thus in what follows, we will let $\mathcal{O}_n$ and $\mathcal{E}_n$ denote the matrices with entries given by the non-monic Chebyshev polynomials \tcnew{of the first kind}.
\\
Let $n\geq 2$ be even. Then using partial fraction decomposition, we have
\begin{align*}\int_{-1}^1 \frac{T_n(x)}{\sqrt{1-x^2}(1+x^2)^2}dx&=\frac{i}{4}\int_{-1}^1 \frac{T_n(x)}{(x+i)\sqrt{1-x^2}}dx-\frac{i}{4}\int_{-1}^1 \frac{T_n(x)}{(x-i)\sqrt{1-x^2}}dx\\&-\frac{1}{4}\int_{-1}^1 \frac{T_n(x)}{(x+i)^2\sqrt{1-x^2}}dx-\frac{1}{4}\int_{-1}^1 \frac{T_n(x)}{(x-i)^2\sqrt{1-x^2}}dx.
\end{align*}
Notice that for $x \notin (-1,1)$,
\begin{equation}\label{eq:main_integral}
\begin{split}
    \int_{-1}^1 \frac{T_n(y)}{(y-x)\sqrt{1-y^2}}dy&= \int_{-1}^1 \frac{T_n(y)-T_n(x)}{y-x}\frac{dy}{\sqrt{1-y^2}}+T_n(x)\int_{-1}^1\frac{1}{y-x}\frac{dy}{\sqrt{1-y^2}}\\
    &=\pi U_{n-1}(x)+T_n(x)\int_{-1}^1\frac{1}{y-x}\frac{dy}{\sqrt{1-y^2}}
    \end{split}
\end{equation} where $U_{n}(x)$ is the $n$-th degree Chebyshev polynomial of the second kind. Differentiating, we see that
\tcnew{,}
\begin{equation}\label{eq:main_integral2}
\begin{split}
    \int_{-1}^1 \frac{T_n(y)}{(y-x)^2\sqrt{1-y^2}}dy &=\pi U'_{n-1}(x)+T'_n(x)\int_{-1}^1\frac{1}{y-x}\frac{dy}{\sqrt{1-y^2}}\\
    &+T_n(x)\int_{-1}^1\frac{1}{(y-x)^2}\frac{dy}{\sqrt{1-y^2}}\tcnew{.}
    \end{split}
    \end{equation}
Thus, using the fact that for $n$ even, \begin{align*}T_n(-i)&=T_n(i)\\
T'_n(-i)&=-T_n(i)\\ U_{n-1}(i)&=-U_{n-1}(i)\\ U'_{n-1}(-i)&=U'_{n-1}(i)
\end{align*}and also using the identities
\begin{equation}\label{eq:identity1}T'_n(i)=nU_{n-1}(i)
\end{equation}
and
\begin{equation}\label{eq:identity2}
    U'_{n-1}(i)=\frac{nT_n(i)-iU_{n-1}(i)}{-2}
\end{equation}
we have that
\begin{equation}\label{eq:integral}\int_{-1}^1\frac{T_n(x)}{(1+x^2)^2\sqrt{1-x^2}}dx = \pi\left(\frac{3+n\sqrt{2}}{4\sqrt{2}}\right)T_n(i)-i\pi\left(\frac{3\sqrt{2}+2n}{4\sqrt{2}} \right)U_{n-1}(i).
\end{equation}
Now, assume by way of contradiction that $\det(\mathcal{E}_n)=0$. Then\tcnew{,}
\[T'_n(i)\int_{-1}^1\frac{T_{n-2}(x)}{(1+x^2)^2\sqrt{1-x^2}}dx=T'_{n-2}(i)\int_{-1}^1\frac{T_n(x)}{(1+x^2)^2\sqrt{1-x^2}}dx
\]so substituting in \eqref{eq:integral} and simplifying, we must have 
\begin{equation}\label{eq:even}
\begin{split}
    n\left( \frac{3+(n-2)\sqrt{2}}{4\sqrt{2}}\right)
    U_{n-1}(i)T_{n-2}(i)&-(n-2)\left( \frac{3+n\sqrt{2}}{4\sqrt{2}}\right)U_{n-3}(i)T_n(i)=\\
    &=\frac{3i}{2}U_{n-1}(i)U_{n-3}(i)\tcnew{.}
    \end{split}
\end{equation} Using the fact that
\begin{equation}\label{eq:identity}
\begin{split}
T_{n-2}(i)U_{n-1}(i)&=\frac{1}{2}\left(U_{2n-3}(i)+2i \right) \text{ and }\\
T_n(i)U_{n-3}(i)&=\frac{1}{2}\left(U_{2n-3}(i)-2i\right)
\end{split}
\end{equation} equation \eqref{eq:even} becomes
\begin{equation}\label{eq:even_2}
    \frac{3}{4\sqrt{2}}U_{2n-3}(i)+\frac{3(n-1)i}{2\sqrt{2}}+\frac{n(n-2)i}{2}= \frac{3i}{2}U_{n-1}(i)U_{n-3}(i).
\end{equation}
Note that 
\[U_n(i)=\frac{(i+\sqrt{2}i)^{n+1}-(i-\sqrt{2}i)^{n+1}}{2\sqrt{2}i}
\] thus, for $n$ even,
\[U_{n-1}(i)U_{n-3}(i)= \frac{1}{8}(3+2\sqrt{2})^{n-1}+\frac{1}{8}(3-2\sqrt{2})^{n-1}-\frac{3}{4}
\] and
\[U_{2n-3}(i)=\frac{-(3+2\sqrt{2})^{n-1}+(3-2\sqrt{2})^{n-1}}{2\sqrt{2}i}
\] so that equation \eqref{eq:even_2} is equivalent to
\begin{equation}\label{eq:even_3}
    \frac{3(n-1)}{2\sqrt{2}}+\frac{n(n-2)}{2}+\frac{9}{8}=\frac{3}{8}(3-2\sqrt{2})^{n-1}.
\end{equation} 
But since $\frac{9}{8}>1$ and all the terms on the left-hand side are positive, we see that the left-hand side is strictly greater than 1. However, $3-2\sqrt{2}<1$ hence $(3-2\sqrt{2})^{n-1}<1$ for all $n\geq 2$ which shows that the right-hand side of \eqref{eq:even_3} is strictly less than 1 which is a contradiction. Therefore, $\det(\mathcal{E}_n)\neq 0$ for $n$ even, $n \geq 2$.

Now let $n$ be odd, $n \geq 1$. Using partial fraction decomposition as before, we have that
\begin{align*}
    \int_{-1}^1\frac{\mathcal{R}_1(x)T_n(x)}{\sqrt{1-x^2}(1+x^2)^2}dx&=\frac{1}{2}\int_{-1}^1\frac{\tcnew{\mathcal{R}_1(x)}T_n(x)}{(x+i)\sqrt{1-x^2}}dx+\frac{1}{2}\int_{-1}^1\frac{\tcnew{\mathcal{R}_1(x)}T_n(x)}{(x-i)\sqrt{1-x^2}}dx\\
    &+\frac{i}{2}\int_{-1}^1\frac{\tcnew{\mathcal{R}_1(x)}T_n(x)}{(x+i)^2\sqrt{1-x^2}}dx-\frac{i}{2}\int_{-1}^1\frac{\tcnew{\mathcal{R}_1(x)}T_n(x)}{(x-i)^2\sqrt{1-x^2}}dx
\end{align*}
so using equations \eqref{eq:main_integral},  \eqref{eq:main_integral2}, \eqref{eq:identity1}, \eqref{eq:identity2} and the fact that for $n$ odd, 
\begin{align*}
    T_n(-i)&=-T_n(i)\\
    T'_n(-i)&=T'_n(i)\\
    U_{n-1}(-i)&=U_{n-1}(i)\\
    U'_{n-1}(-i)&=-U'_{n-1}(i)
\end{align*} we see that
\begin{align*}
    \int_{-1}^1\frac{\mathcal{R}_1(x)T_n(x)}{\sqrt{1-x^2}(1+x^2)^2}dx&=\left( \frac{3i\pi}{2\sqrt{2}}+\frac{i\pi n}{2}\right)T_n(i) +\left( \frac{\pi n}{\sqrt{2}}+\frac{3\pi}{2}\right)U_{n-1}(i).
\end{align*}
As before, assume that $\det(\mathcal{O}_n)=0$. Then 
\[T'_n(i)\int_{-1}^1\frac{(x^3+3x)T_{n-2}(x)}{(1+x^2)^2\sqrt{1-x^2}}dx=T'_{n-2}(i)\int_{-1}^1\frac{(x^3+3x)T_{n}(x)}{(1+x^2)^2\sqrt{1-x^2}}dx
\] hence,
\begin{equation}
    \begin{split}
    n\left(\frac{3i\pi}{2\sqrt{2}}+\frac{i\pi(n-2)}{2} \right)U_{n-1}(i)T_{n-2}(i)
    &-(n-2)\left(\frac{3i\pi}{2\sqrt{2}}+\frac{i\pi n}{2} \right)U_{n-3}(i)T_n(i)=\\
    &=3\pi U_{n-1}(i)U_{n-3}(i).
    \end{split}
\end{equation}
By equation \eqref{eq:identity}, this simplifies to 
\begin{equation}\label{eq:odd}
    \frac{3i}{2\sqrt{2}}U_{2n-3}(i)-\frac{3(n-1)}{\sqrt{2}}-n(n-2)=-3U_{n-1}(i)U_{n-3}(i).
\end{equation}
Since $n$ is odd, 
\[U_{n-1}(i)U_{n-3}(i) = \frac{-1}{8}\left((3+2\sqrt{2})^{n-1}-(3-2\sqrt{2})^{n-1}+6 \right)
\] and
\[U_{2n-3}(i)=\frac{(3+2\sqrt{2})^{n-1}-(3-2\sqrt{2})^{n-1}}{2\sqrt{2}i}
\] so that equation \eqref{eq:odd} is equivalent to
\begin{equation}
    \frac{-3(n-1)}{\sqrt{2}}-(n-2)-\frac{9}{4}=\frac{3}{4}(3-2\sqrt{2})^{n-1}. 
\end{equation}
Clearly, for $n \geq 1$, the left-hand side is less than zero but the right-hand side is greater than 0, hence $\det(\mathcal{O}_n)\neq 0$ for any $n$ odd $n\geq 1$.
\end{proof}

We list the first few $\mathcal{R}_n(x)$ below:
\begin{align*}
    \mathcal{R}_0(x)&=1\\
    \mathcal{R}_1(x)&=x^3+3x\\
    \mathcal{R}_2(x)&=x^4+2x^2+1-\frac{4\sqrt{2}}{3}\\
    \mathcal{R}_3(x)&=x^5+\frac{41-5\sqrt{2}}{28}x^3+\frac{-17-15\sqrt{2}}{28}x\\
    \mathcal{R}_4(x)&=x^6-\frac{3(-859+192\sqrt{2})}{2402}x^4-\frac{2052+1152\sqrt{2}}{2402}x^2-\frac{7859+5592\sqrt{2}}{2402}
\tcnew{.}    
\end{align*}

\tcnew{Next, using Mathematica one can check the behavior of the zeros of $\mathcal{R}_n$.}
      \begin{figure}[h!]\centering
\begin{minipage}{.5\textwidth}
  \centering
   \includegraphics[scale=0.5]{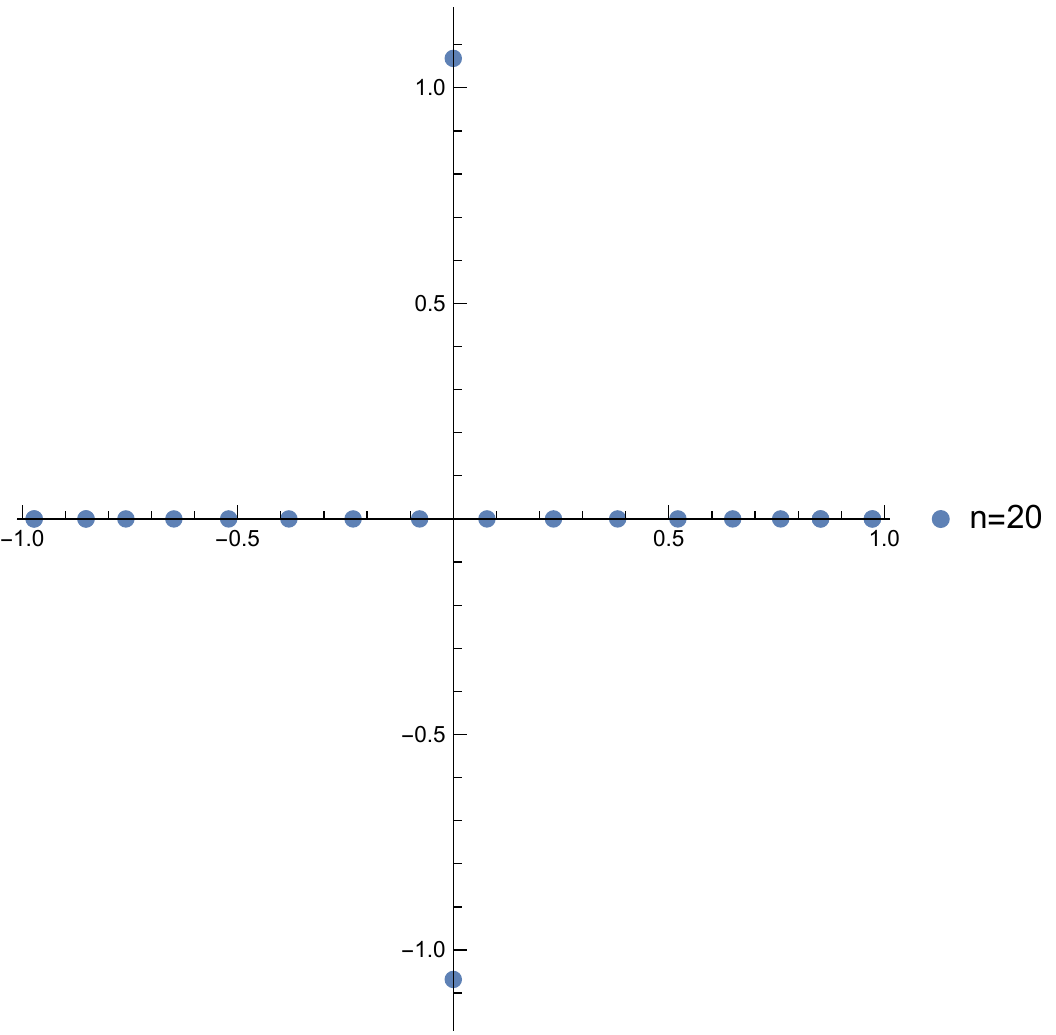}
\end{minipage}%
\begin{minipage}{.5\textwidth}
  \centering
  \includegraphics[scale=0.5]{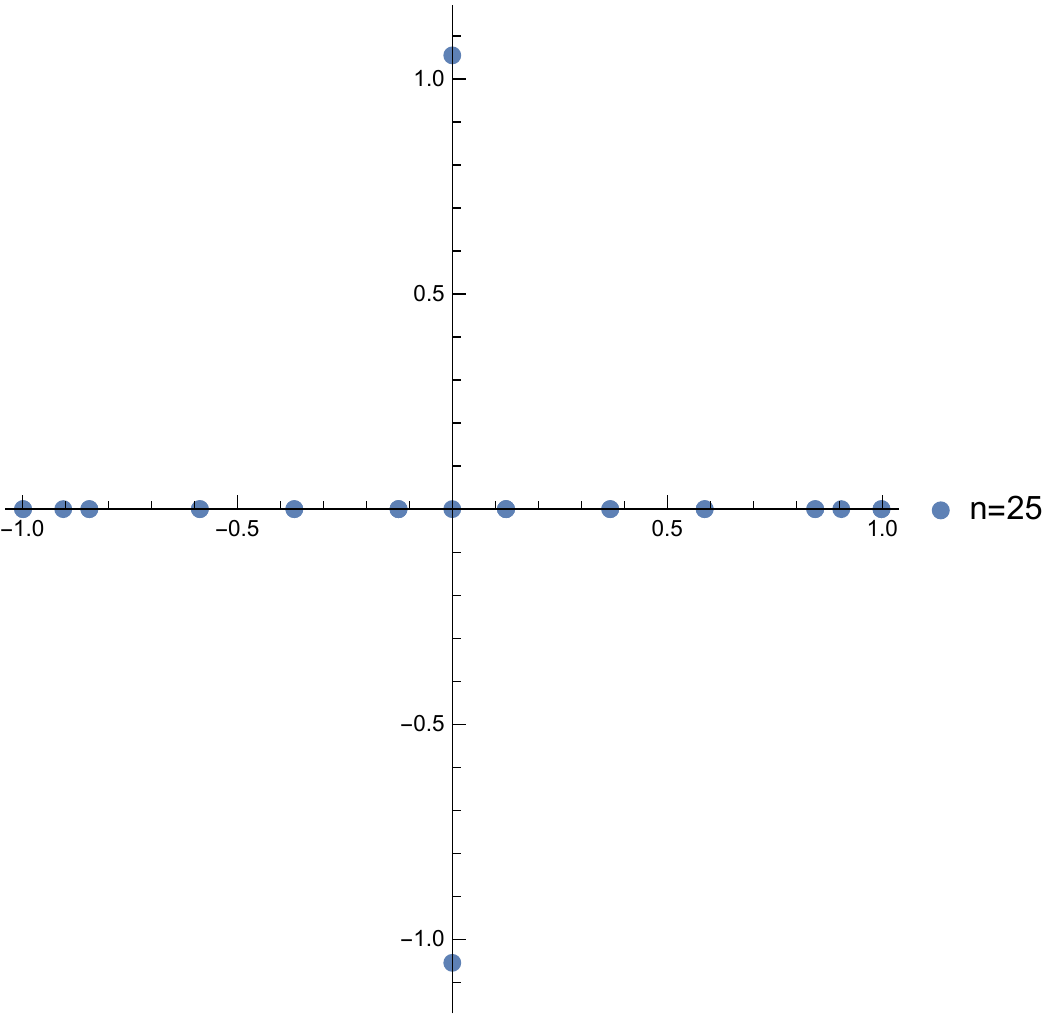}
\end{minipage}
\caption{\tcnew{The behavior of the zeros of $\mathcal{R}_n(x)$ for $n=20$ and $n=25$}. \tcnew{Note that for $n=20$, there are $4$ real zeros of multiplicity two and for} \tcnew{$n=25$, there are $7$ with multiplicity two and two with} \tcnew{multiplicity $3$}.}
\label{fig:RZeros}
    \end{figure}
\tcnew{
\begin{remark}
    From Figure 1, one can see that the zeros of $\mathcal{R}_n$ behave similarly to zeros of exceptional Hermite, Jacobi, and Laguerre polynomials \cite{GMM}, \cite{KM15}. Besides, such a behavior is typical for indefinite orthogonal polynomials \cite{DD04}, \cite{DD07}. Since the link between general $R_n$ and indefinite orthogonal polynomials has been given earlier in this paper, it is natural to conjecture that similar situation takes place for general polynomials $R_n$.
\end{remark}
}
    
\begin{corollary} Let $\{\mathcal{R}_n(x)\}$ be the family of polynomials given in Theorem \ref{Thm:ExCheb}. Then $\{\mathcal{R}_n(x)\}$ is a family of exceptional orthogonal polynomials and $\mathcal{R}_n(i)\neq 0$ for any $n=0,1,2,\dots$.
\end{corollary}
\begin{proof}
By Theorem \ref{Thm:ExCheb}, we know that the $\mathcal{R}_n(x)$ are polynomials of degree $n+2$\tcnew{,} thus, the sequence does not include a degree 1 or degree 2 polynomial. The orthogonality and \tcnew{the} fact that $\mathcal{R}_n(i)\neq 0$ for any $n$ follow from Theorem \ref{thrm:Rn} and Proposition \ref{Rni}, respectively.
\end{proof}
The results of Theorem \ref{Thm:ExCheb} allow us to apply Theorem \ref{GenChristoffel} to obtain the monic Chebyshev polynomials \tcnew{of the first kind} from the $\mathcal{R}_n(x)$.
\begin{corollary}
Let $\{\mathcal{R}_n(x)\}$ be the family of polynomials given in Theorem \ref{Thm:ExCheb} and let $\{\hat{T}_n(x)\}$ be the sequence of monic Chebyshev polynomials \tcnew{of the first kind}. Then for $\phi(x)=(1+x^2)^2$, we have
\begin{multline*}
    \phi(x)\hat{T}_n(x)= \\
= \left[ \frac{1}{c_n}\begin{vmatrix}\mathcal{R}_n(i)&\mathcal{R}_{n+1}(i)&\mathcal{R}_{n+2}(i)\\
\mathcal{R}_n(-i)&\mathcal{R}_{n+1}(-i)&\mathcal{R}_{n+2}(-i)\\
\mathcal{R}_n(x)&\mathcal{R}_{n+1}(x)&\mathcal{R}_{n+2}(x)
\end{vmatrix}+\frac{\rho_n}{c_{n-2}}\begin{vmatrix}\mathcal{R}_{n-2}(i)&\mathcal{R}_{n-1}(i)&\mathcal{R}_{n}(i)\\
\mathcal{R}_{n-2}(-i)&\mathcal{R}_{n-1}(-i)&\mathcal{R}_{n}(-i)\\
\mathcal{R}_{n-2}(x)&\mathcal{R}_{n-1}(x)&\mathcal{R}_{n}(x)
\end{vmatrix}\right]
\end{multline*}
where\tcnew{,}
\[c_n=\begin{vmatrix}
\mathcal{R}_n(i)&\mathcal{R}_{n+1}(i)\\
\mathcal{R}_n(-i) & \mathcal{R}_{n+1}(-i)
\end{vmatrix}
\]
and $\{\rho_n\}$ is a sequence of real numbers.
\end{corollary}

\begin{proof}
This is a direct application of Theorem \ref{Thm:ExCheb}. 
\end{proof}
One may note that the $\rho_n$ are given by Theorem \ref{thm:Rho_n}\tcnew{,} where \newline $d\mu(x) = \frac{1}{\sqrt{1-x^2}}\chi_{[-1,1]}(x)dx$ and \[S_n(x)= \frac{1}{c_n\phi(x)}\begin{vmatrix}\mathcal{R}_n(i)&\mathcal{R}_{n+1}(i)&\mathcal{R}_{n+2}(i)\\
\mathcal{R}_n(-i)&\mathcal{R}_{n+1}(-i)&\mathcal{R}_{n+2}(-i)\\
\mathcal{R}_n(x)&\mathcal{R}_{n+1}(x)&\mathcal{R}_{n+2}(x)
\end{vmatrix}.
\] Below are the first few $S_n(x)$ corresponding to the Chebyshev polynomials:
\begin{align*}
    S_0(x)&=1\\
    S_1(x)&=x\\
    S_2(x)&=x^2-\frac{1}{2}\\
    S_3(x)&=x^3-\frac{23}{30}x\\
    S_4(x)&=x^4-\frac{49}{48}x^2+\frac{13}{96}\tcnew{.}
\end{align*}
One can also quickly check that, for example, when $n=2$, there are no such $\alpha$ and $\beta$ such that
\[xS_2(x)=S_3(x)+\alpha S_2(x)+\beta S_1(x)
\] thus the $S_n(x)$ corresponding to the Chebyshev polynomials \tcnew{of the first kind} do not form an orthogonal polynomial sequence with respect to any quasi-definite linear functional. 

Below we illustrate the modification of the Christoffel formula when applied to the $\mathcal{R}_n(x)$.
\begin{example} Applying Theorem \ref{GenChristoffel} to the $\mathcal{R}_n(x)$ for $n=0,1,2,3,4$, we have
\begin{align*}
\phi(x)\hat{T}_0(x)&=\phi(x)S_0(x)=\mathcal{R}_2(x)+\frac{4\sqrt{2}}{3}\mathcal{R}_0(x)\\
\phi(x)\hat{T}_1(x)&=\phi(x)S_1(x)=\mathcal{R}_3(x)-\frac{5}{12-4\sqrt{2}}\mathcal{R}_1(x)\\
\phi(x)\hat{T}_2(x)&=\phi(x)S_2(x)=\mathcal{R}_4(x)+\frac{9}{57-32\sqrt{2}}\mathcal{R}_2(x)\\
\phi(x)\hat{T}_3(x)&=\phi(x)\left[ S_3(x)+\frac{1}{60}S_1(x)\right]=\mathcal{R}_5(x)+\frac{50+29\sqrt{2}}{100}\mathcal{R}_3(x)+\frac{3+\sqrt{2}}{336}\mathcal{R}_1(x)\\
\phi(x)\hat{T}_4(x)&=\phi(x)\left[S_4(x)+\frac{1}{48}S_2(x)\right]=\mathcal{R}_6(x)+\frac{2(32-27\sqrt{2})}{832-597\sqrt{2}}\mathcal{R}_4(x)+\\&+\frac{3(57+32\sqrt{2})}{19216}\mathcal{R}_2(x)\tcnew{.}
\end{align*}
\end{example}

\begin{theorem}\label{thm:complete}
The DEK-type polynomials $\mathcal{R}_n(x)$ corresponding to the monic \newline Chebyshev polynomials are complete in $L^2\left( \frac{dx}{\sqrt{1-x^2}(1+x^2)^2}, [-1,1]\right)$.
\end{theorem}
\begin{proof}
Assume that $\{\mathcal{R}_n(x)\}$ is not complete in $L^2\left( \frac{dx}{\sqrt{1-x^2}(1+x^2)^2}, [-1,1]\right)$. \newline Then there exists $f(x) 
\in L^2\left( \frac{dx}{\sqrt{1-x^2}(1+x^2)^2}, [-1,1]\right)$ such that $f(x) \not \equiv 0$ and
\[\int_{-1}^{1}f(x)\mathcal{R}_n(x) \frac{dx}{\sqrt{1-x^2}(1+x^2)^2}=0
\] for all $n=0,1,2\dots$. By Theorem \ref{GenChristoffel}, we have for all $n=1,2,\dots$
\[
(1+x^2)^2\hat{T}_n(x) = \mathcal{R}_{n+2}(x)+\alpha_n\mathcal{R}_{n}(x)+\beta_n\mathcal{R}_{n-2}(x)
\]for real numbers $\alpha_n$ and $ \beta_n$. Thus for all $n=1,2,\dots$
\begin{align*}
\int_{-1}^1f(x)\hat{T}_n(x) \frac{dx}{\sqrt{1-x^2}}&=
\int_{-1}^1f(x)(1+x^2)^2\hat{T}_n(x) \frac{dx}{\sqrt{1-x^2}(1+x^2)^2}\\
&=\int_{-1}^1 f(x)\mathcal{R}_{n+2}(x)\frac{dx}{\sqrt{1-x^2}(1+x^2)^2}\\
&+\alpha_n\int_{-1}^1\tcnew{f(x)}\mathcal{R}_{n}(x)\frac{dx}{\sqrt{1-x^2}(1+x^2)^2}\\
&+\beta_n\int_{-1}^1\tcnew{f(x)}\mathcal{R}_{n-2}(x)\frac{dx}{\sqrt{1-x^2}(1+x^2)^2}\\
&=0.
\end{align*}
If $n=0$ then $(1+x^2)^2\hat{T}_0(x)=(1+x^2)^2S_0(x)=\mathcal{R}_2(x)+\alpha_0\mathcal{R}_0(x)$ so by the same reasoning
\[
\int_{-1}^1f(x)\hat{T}_0(x) \frac{dx}{\sqrt{1-x^2}}=0.
\]

Notice that 
$
\frac{1}{(1+x^2)^2}\geq \frac{1}{4}$ on $[-1,1]$ hence
\begin{align*}
\frac{1}{4}\int_{-1}^1|f(x) |^2\frac{dx}{\sqrt{1-x^2}} &\leq  \int_{-1}^{1}|f(x)|^2 \frac{dx}{\sqrt{1-x^2}(1+x^2)^2}\\
&<\infty
\end{align*} thus $f(x)\in L^2\left( \frac{dx}{\sqrt{1-x^2}} , [-1,1]\right)$. But since $\hat{T}_n(x)$ are complete in this space, the above implies that $f(x)\equiv 0$ which is a contradiction.
\end{proof}
\begin{remark}
Theorem \ref{thm:complete} can be generalized to DEK-type polynomials where the $P_n(x)$ are complete in $L^2\left(d\mu(x), [a,b] \right)$ for a compact subset $[a,b]$ of $\mathbb{R}$.
\end{remark}

\end{subsection}
\end{section}

\vspace{7mm}

\noindent {\bf Acknowledgments.} The authors would like to thank Alberto Gr\"unbaum for stimulating discussions and helpful correspondence \tcnew{and Juan Carlos Garc\'ia-Ardila for a careful reading of the manuscript and pointing out typos. Also, the authors are extremely grateful to the anonymous referees for helpful suggestions and remarks that improved the content and presentation of the paper.} This research was supported by the NSF DMS grant 2008844 and in part by the University of Connecticut Research Excellence Program.

\end{document}